%% file: main.tex
\documentclass{article}
\usepackage[utf8]{inputenc}
\usepackage{amsthm}
\usepackage{amsmath}
\usepackage{amsfonts}
\usepackage{fullpage}
\usepackage{xcolor}
\usepackage{verbatim}
\usepackage{mathtools}
\usepackage{todonotes}
\usepackage[linesnumbered, ruled]{algorithm2e}
\usepackage{tikz}
\usepackage{authblk}
\usepackage{bbm}
\usepackage{hyperref}

\author{Haris Angelidakis}
\author{Dylan Hyatt-Denesik}
\author{Laura Sanit{\`{a}}}
\affil{TU Eindhoven}

\title{Node Connectivity Augmentation via Iterative Randomized Rounding}
\date{}

\newtheorem{theorem}{Theorem}
\newtheorem{lemma}{Lemma}

\newtheorem{definition}{Definition}
\newtheorem{observation}{Observation}
\newtheorem{corollary}{Corollary}
\newtheorem{fact}{Fact}

\newcommand{\E}{\mathbf{E}}
\newcommand{\OPT}{\mathtt{OPT}}
\newcommand{\LP}{\mathrm{LP}}

\begin{document}

\maketitle

\begin{abstract}
Many fundamental network design problems deal with the design of low-cost networks that are resilient to the failure of their elements (such as nodes or links). One such problem is Connectivity Augmentation, where the goal is to cheaply increase the connectivity of a given network from a value $k$ to $k+1$. 

The most studied setting focuses on edge-connectivity, and it is well-known that it can be reduced to the case $k=2$, called Cactus Augmentation. From an approximation perspective, Byrka, Grandoni, and Jabal Ameli (2020) were the first who managed to break the 2-approximation barrier for this problem, by exploiting a connection to the Steiner Tree problem, and by tailoring the analysis of the iterative randomized rounding technique for Steiner Tree to the specific instances arising from this connection. Very recently, Nutov (2020) observed that a similar reduction to Steiner Tree also holds for a node-connectivity problem called Block-Tree Augmentation, where the goal is to add edges to a given spanning tree in order to make the resulting graph 2-node-connected. Block-Tree Augmentation can be seen as the direct generalization of the well-studied Tree Augmentation problem ($k=1$) to the node-setting. Combining Nutov's result with the algorithm of Byrka, Grandoni, and Jabal Ameli yields a 1.91-approximation for Block-Tree Augmentation, that is the best bound known so far.

In this work, we give a 1.892-approximation algorithm for the problem of augmenting the node-connectivity of any given graph from 1 to 2. As a corollary, we improve upon the state-of-the-art approximation factor for Block-Tree Augmentation. Our result is obtained by developing a different and simpler analysis of the iterative randomized rounding technique when applied to the Steiner Tree instances arising from the aforementioned reductions. Our results also imply a 1.892-approximation algorithm for Cactus Augmentation. While this does not beat the best approximation factor by Cecchetto, Traub, and Zenklusen (2021) that is known for this problem, a key point of our work is that the analysis of our approximation factor is quite simple compared to previous results in the literature. In addition, our work gives new insights on the iterative randomized rounding method, that might be of independent interest.
\end{abstract}

\input{./intro.tex}

\input{./dcr_relaxation.tex}

\input{./alg.tex}

\input{./saving_tree.tex}

\input{./leaf_adjacent_block_tap.tex}

\input{./lower_bound.tex}

\bibliographystyle{plain}
\bibliography{references}
\appendix
\input{./appendix.tex}

\end{document}

%% file: intro.tex
\section{Introduction}

Survivable network design deals with the design of a low-cost network that is resilient to the failure of some of its elements (like nodes or links), and finds application in various settings such as, e.g., telecommunications and transportation. One of the most fundamental problems in this area is Connectivity Augmentation, which asks to cheaply increase the connectivity of a given network from a value $k$ to $k+1$. The most classical setting considers \emph{edge}-connectivity: here the input consists of a $k$-edge-connected graph $G$ and a set $L$ of extra links, and the goal is to select the smallest subset $L' \subseteq L$ of links such that if we add $L'$ to graph $G$, then its connectivity increases by 1, i.e., it becomes $(k+1)$-edge-connected. This Edge-Connectivity Augmentation problem has a long history. It was observed long ago (see Dinitz et al.~\cite{Dinic76}, as well as Cheriyan et al.~\cite{DBLP:conf/esa/CheriyanJR99} and Khuller and Thurimella~\cite{DBLP:journals/jal/KhullerT93}) that there is an approximation-preserving reduction from the Edge-Connectivity Augmentation problem for an arbitrary $k$ to the case where $k = 1$, if $k$ is odd, and to the case where $k = 2$, if $k$ is even. When $k = 1$, the problem is known as the Tree Augmentation problem (TAP), since the input graph in that case can be assumed to be a tree without loss of generality, while the case of $k = 2$ is called the Cactus Augmentation problem (CacAP), since the input graph in that case can be assumed to be a cactus graph\footnote{A cactus is a connected graph where every edge is part of exactly one cycle.}.
It is easy to see that TAP is a special case of CacAP\footnote{For that, we can simply add two parallel copies of each edge in the input tree of a TAP instance and the resulting instance is equivalent to a CacAP instance.}. This, combined with the aforementioned reduction, implies that an $\alpha$-approximation algorithm for CacAP yields an $\alpha$-approximation algorithm for the general Edge-Connectivity Augmentation problem. Both TAP and CacAP admit simple 2-approximation algorithms~\cite{DBLP:conf/soda/GoemansGPSTW94,DBLP:journals/combinatorica/Jain01,DBLP:journals/jal/KhullerT93}. Approximation algorithms with approximation ratio better than $2$ have been discovered for TAP in a long line of research spanning three decades~\cite{DBLP:journals/talg/Adjiashvili19,DBLP:journals/corr/abs-2012-00086,DBLP:journals/algorithmica/CheriyanG18,DBLP:journals/algorithmica/CheriyanG18a,DBLP:journals/orl/CheriyanKKK08,DBLP:journals/tcs/CohenN13,DBLP:journals/talg/EvenFKN09,DBLP:conf/soda/Fiorini0KS18,DBLP:journals/siamcomp/FredericksonJ81,DBLP:conf/stoc/0001KZ18,DBLP:journals/jal/KhullerT93,DBLP:journals/talg/KortsarzN16,DBLP:journals/dam/KortsarzN18,DBLP:journals/dam/Nagamochi03,DBLP:conf/esa/Nutov17}, and more recently for CacAP~\cite{DBLP:conf/stoc/Byrka0A20,DBLP:journals/corr/abs-2012-00086,DBLP:journals/corr/abs-2009-13257}.

The above problems naturally extend to their node-connected variants. However, approximation results on Node-Connectivity Augmentation are more scarce, even for the most basic generalization, known as the Block-Tree Augmentation problem (Block-TAP), which is the direct extension of TAP to the node-connected case. We now define it formally.

\begin{definition}[Block-TAP]
Let $T=(V, E)$ be a tree, and $L \subseteq \binom{V}{2}$. The goal is to compute a minimum cardinality set $L'\subseteq L$ such that $G = (V, E \cup L')$ is a $2$-node-connected graph.
\end{definition}

Similar to CacAP, a $2$-approximation algorithm was known for quite some time~\cite{DBLP:journals/siamcomp/FredericksonJ81,DBLP:journals/jal/KhullerT93} for Block-TAP, and until recently, it was an open question to design an approximation algorithm with ratio better than $2$. Very recently, Nutov~\cite{nutov20202nodeconnectivity} observed that Block-TAP can be reduced to special instances of the (unweighted) Node Steiner Tree problem, extending the techniques of Basavaraju et al.~\cite{DBLP:conf/icalp/BasavarajuFGMRS14} and Byrka et al.~\cite{DBLP:conf/stoc/Byrka0A20} used for CacAP. We recall that the (unweighted) Node Steiner Tree problem takes as input a graph $G$ and a subset of nodes (called terminals), and asks to find a tree spanning the terminals which minimizes the number of non-terminal nodes included in the tree. These special Node Steiner Tree instances exhibit some crucial properties, similar to the Steiner Tree instances constructed by~\cite{DBLP:conf/stoc/Byrka0A20} for CacAP, and hence the $1.91$-approximation of~\cite{DBLP:conf/stoc/Byrka0A20} for CacAP also yields a $1.91$-approximation to Block-TAP. This is the first result breaking the barrier of $2$ on its approximability, and the best bound known so far.

\paragraph{Our results and techniques.} In this work, we consider the problem of augmenting the node-connectivity of a graph from 1 to 2, that is slightly more general than Block-TAP.

\begin{definition}[1-Node-CAP]
Let $G=(V, E)$ be a $1$-node-connected graph, and $L \subseteq \binom{V}{2}$. The goal is to compute a minimum cardinality set $L'\subseteq L$ such that $G' = (V, E \cup L')$ is a $2$-node-connected graph.
\end{definition}

Our main result is the following.
\begin{theorem}\label{thm:approx-1node-tap}
There exists a $1.892$-approximation algorithm for 1-Node-CAP.
\end{theorem}

Note that Block-TAP is a particular case of 1-Node-CAP where the input graph $G$ is a tree\footnote{It is stated in~\cite{nutov20202nodeconnectivity} that 1-Node-CAP and Block-TAP are equivalent by relying on constructing the so-called block-cut tree of a given graph. While this is immediate if one allows \emph{weights} for the links (see also~\cite{DBLP:journals/siamcomp/FredericksonJ81,DBLP:journals/jal/KhullerT93}), the same reduction does not seem to hold in the unweighted setting. We refer to~\ref{appendix:1node-block-tap} for more details.}. Therefore as a corollary, with our result we improve upon the approximation bound for Block-TAP.

\begin{corollary}\label{thm:approx-block-tap}
There exists a $1.892$-approximation algorithm for Block-TAP.
\end{corollary}

Moreover, as another corollary of our techniques and prior results, we also get the following result about CacAP.
\begin{theorem}\label{thm:approx-cacap}
There exists a $1.892$-approximation algorithm for CacAP.
\end{theorem}

The starting point of our work is the reduction of Block-TAP and CacAP to the previously mentioned Node Steiner Tree instances of~\cite{DBLP:conf/icalp/BasavarajuFGMRS14,DBLP:conf/stoc/Byrka0A20,nutov20202nodeconnectivity}, which from now on we call CA-Node-Steiner-Tree instances (formally defined in Definition~\ref{def:ca-node-steiner}). We first observe that instances of the more general 1-Node-CAP can be reduced to CA-Node-Steiner-Tree instances. We clarify that these instances, obtained by the above reduction, can also be treated as Edge Steiner Tree instances (as done by~\cite{DBLP:conf/stoc/Byrka0A20}) but we will view them as Node Steiner Tree instances, since this allows for a more direct correspondence between ``links to add'' for 1-Node-CAP/Block-TAP/CacAP, and ``Steiner nodes to select'' for Steiner Tree. This view helps us to give a cleaner analysis of the iterative randomized rounding technique. Thus, the main task for proving Theorems~\ref{thm:approx-1node-tap} and~\ref{thm:approx-cacap} is to design a $1.892$-approximation for CA-Node-Steiner-Tree.

Besides giving a 1.892 approximation for 1-Node-CAP, which improves upon the state-of-the-art approximation of $1.91$ by Nutov~\cite{nutov20202nodeconnectivity} for Block-TAP, one key point of our work is that the analysis of our approximation bound is quite simple compared to the existing results in the literature for CacAP that achieve a better than 2 approximation. More precisely, only the algorithm of~\cite{DBLP:journals/corr/abs-2012-00086} gives a better approximation factor than our algorithm---in fact, much better ---but its analysis is far more involved, spanning over 70 pages. Furthermore, our work gives some new insights on the iterative randomized rounding method introduced by Byrka et al.~\cite{DBLP:journals/jacm/ByrkaGRS13} that might be of independent interest. We give a few more details of this iterative rounding next.

The iterative randomized rounding technique, applied to CA-Node-Steiner-Tree, at each iteration uses an (approximate) LP-relaxation to sample a set of Steiner nodes connecting part of the terminals, contract them, and iterate until all terminals are connected. Roughly speaking, the heart of the analysis lies in bounding the expected number of iterations of the algorithm until a Steiner node of a given initial optimal solution is not needed anymore in the current (contracted) instance. This is achieved by a suitably chosen spanning tree on the set of terminals (called the \emph{witness tree}). In the original Steiner Tree work given by~\cite{DBLP:journals/jacm/ByrkaGRS13}, as well as in the work by~\cite{DBLP:conf/stoc/Byrka0A20}, the witness tree is chosen (mostly) randomly. Later an edge-deletion process over this tree is mapped to an edge-deletion process over the edges of an optimal solution. In contrast, we give a \emph{purely deterministic} way to construct the witness tree, and then map an edge-deletion process over this tree to a node-deletion process over the Steiner nodes of an optimal solution. The deterministic method of constructing the witness tree that we introduce here relies on computing some \emph{minimum weight paths} from the Steiner nodes in an optimal solution to terminals, according to some node-weights that take into account two factors: the number of nodes in the path, and the degree of each internal node in the path.

Moreover, our techniques can be refined and give a $(1.8\bar{3} + \varepsilon)$-approximation algorithm for what we call leaf-adjacent Block-TAP instances; these are Block-TAP instances where at least one endpoint of each link is a leaf. We note here that Nutov~\cite{DBLP:journals/corr/abs-2009-13257} recently gave a $1.6\bar{6}$-approximation for leaf-to-leaf Block-TAP instances; there are instances where both endpoints of each link are leaves. Thus, our $(1.8\bar{3} + \varepsilon)$-approximation algorithm deals with a strictly larger set of instances compared to~\cite{DBLP:journals/corr/abs-2009-13257}, albeit with a worse approximation factor.

\begin{theorem}\label{thm:approx-leaf-adjacent-block-tap}
For any fixed $\varepsilon >0$, there exists a $(1.8\bar{3} + \varepsilon)$-approximation algorithm for leaf-adjacent Block-TAP.
\end{theorem}

Interestingly, we can also provide some concrete limits on how much our techniques can be further pushed. We show that there exist leaf-to-leaf Block-TAP instances for which our choice of the witness tree is the \emph{best} possible, and yields a tight bound of $1.8\bar{3}$. This shows that any approximation bound strictly better than $1.8\bar{3}$ needs substantially different arguments.

%

\paragraph{Related work.} We clarify here that the $2$-approximation algorithms mentioned for TAP, CacAP, Block-TAP work even with \emph{weighted} instances, where each link comes with a non-negative cost and the goal is to minimize the overall cost of the selected links. In fact, the $2$-approximation algorithms for Block-TAP~\cite{DBLP:journals/siamcomp/FredericksonJ81,DBLP:journals/jal/KhullerT93} works more generally for (weighted) 1-Node-CAP. Differently, all other aforementioned algorithms only apply to the unweighted version of TAP, CacAP and Block-TAP (some of them extend to the weighted version with bounded integer weights). Traub and Zenklusen~\cite{DBLP:journals/corr/abs-2104-07114,DBLP:journals/corr/abs-2107-07403} very recently presented a $1.7$-approximation algorithm and a $(1.5 + \varepsilon)$-approximation algorithm for the general weighted version of TAP, thus breaking the long-standing barrier of $2$ for weighted TAP. For more details about the history of these problems, we refer the reader to~\cite{DBLP:journals/corr/abs-2012-00086,DBLP:conf/stoc/0001KZ18} and the references contained in them.

Regarding CacAP, last year Byrka et al.~\cite{DBLP:conf/stoc/Byrka0A20} managed to break the $2$-approximation barrier and obtained a $1.91$-approximation, thus resulting in the first algorithm with an approximation factor strictly smaller than $2$ for the general Edge-Connectivity Augmentation problem. As mentioned, to get this result, they exploited a reduction of Basavaraju et al.~\cite{DBLP:conf/icalp/BasavarajuFGMRS14} to the (unweighted) Edge Steiner Tree problem and utilized the machinery developed for that problem by Byrka et al.~\cite{DBLP:journals/jacm/ByrkaGRS13}. We recall here that the (unweighted) Edge Steiner Tree problem takes as input a graph $G$ and a subset of nodes (called terminals), and asks for a minimum size subtree of $G$ spanning the terminals. Specifically, Byrka et al.~\cite{DBLP:conf/stoc/Byrka0A20} tailored the analysis of the iterative randomized rounding algorithm for Steiner Tree in~\cite{DBLP:journals/jacm/ByrkaGRS13} to the specific Steiner Tree instances arising from the reduction. Nutov~\cite{DBLP:journals/corr/abs-2009-13257} then showed how one can use any algorithm for the Edge Steiner Tree problem in a black-box fashion in order to obtain approximation algorithms for CacAP, at the expense of a slightly worse approximation bound; in particular, he obtained a $1.942$-approximation by applying the algorithm of~\cite{DBLP:journals/jacm/ByrkaGRS13} as a black box.
Soon after, Cecchetto et al.~\cite{DBLP:journals/corr/abs-2012-00086} (relying on a completely different approach, more in line with the techniques used for TAP) gave a very nice unified algorithm that gives a state-of-the-art 1.393-approximation for CacAP, and thus, also for TAP and the general Edge-Connectivity Augmentation problem.

 Regarding Node-Connectivity Augmentation for $k > 1$, we refer the reader to~\cite{DBLP:journals/siamcomp/CheriyanV14,DBLP:conf/soda/Nutov20} and the references contained in them.

\paragraph{Organization of material.} The rest of this work is organized as follows. In Section 2, we start by formally defining CA-Node-Steiner-Tree instances, and state the connection to CacAP, Block-TAP, and 1-Node-CAP. We also explicitly (re)prove that CA-Node-Steiner-Tree instances admit a hypergraph LP-relaxation which is based on so-called $k$-restricted components. We then describe the iterative randomized rounding procedure for these node-based instances in Section 3. Some missing proofs from Sections 2 and 3 can be found in Appendix~\ref{appendix:proofs}. Section 4 contains our simpler witness tree analysis. Finally, Section~\ref{sec:leaf-adjacent-upper-bound} reports the improved approximation for leaf-adjacent Block-TAP instances, while Section~\ref{sec:lower-bound} contains lower bound constructions that demonstrate the limits of our techniques.

%% file: dcr_relaxation.tex
\section{From Connectivity Augmentation to Node Steiner Tree}

\subsection{Reduction to CA-Node-Steiner-Tree instances}


As a reminder, in the Node Steiner Tree problem, we are given a graph $G = (V, E)$ and a subset $R\subseteq V$ of nodes, called \emph{terminals}, and the goal is to compute a tree $T$ of $G$ that contains all the terminals and minimizes the number of non-terminal nodes (so-called \emph{Steiner nodes}) contained in $T$. We will refer to the number of Steiner nodes contained in a tree $T$ as the \emph{cost} of the solution, and indicate it with $cost(T)$.

In general, the Node Steiner Tree problem is as hard to approximate as the Set Cover problem, and it admits a $O(\log |R|)$ approximation algorithm (which holds even in the more general weighted version~\cite{KLEIN1995104}). However, the instances that arise from the reductions of~\cite{DBLP:conf/icalp/BasavarajuFGMRS14,DBLP:conf/stoc/Byrka0A20,nutov20202nodeconnectivity} have special properties that allow for a constant factor approximation. We now define the properties of these instances. We use the notation $N_G(u)$ to denote the set of nodes that are adjacent to a node $u$ in a graph $G$ ($u$ is not included in $N_G(u)$).

\begin{definition}[CA-Node-Steiner-Tree]
\label{def:ca-node-steiner}
Let $G=(V,E)$  be an instance of Node Steiner Tree, with $R\subseteq V$ being the set of terminals. The instance $G$ is a \textbf{CA-Node-Steiner-Tree} instance if the following hold:
\hfill
\begin{enumerate}
    \item For each terminal $v \in R$, we have $N_G(v) \cap R = \emptyset$.
    \item For each Steiner node $\ell \in V \setminus R$, we have $|N_G(\ell) \cap R| \leq 2$.
    \item For each terminal $v \in R$, the set $N_G(v)$ forms a clique in $G$.
\end{enumerate}
\end{definition}

As already mentioned, the starting point of this work is the reduction from CacAP and Block-TAP to CA-Node-Steiner-Tree. An overview of the reduction from CacAP to CA-Node-Steiner-Tree by Byrka et al.~\cite{DBLP:conf/stoc/Byrka0A20} is found in Appendix~\ref{appendix:approx-preserve}. In Appendix~\ref{appendix:approx-preserve-block} we extend the reduction from Block-TAP to CA-Node-Steiner-Tree by Nutov~\cite{nutov20202nodeconnectivity}, to deal with 1-Node-CAP instances. The following two theorems make these connections formal.

\begin{theorem}[\cite{DBLP:conf/stoc/Byrka0A20}]
\label{cor:approx-preserve}
The existence of an $\alpha$-approximation algorithm for CA-Node-Steiner-Tree implies the existence of an $\alpha$-approximation algorithm for CacAP.
\end{theorem}

\begin{theorem}[Extending~\cite{nutov20202nodeconnectivity}]
\label{cor:approx-preserve-block}
The existence of an $\alpha$-approximation algorithm for CA-Node-Steiner-Tree implies the existence of an $\alpha$-approximation algorithm for 1-Node-CAP (and hence Block-TAP).
\end{theorem}

The above two theorems imply that from now on we can therefore concentrate on the CA-Node-Steiner-Tree problem.

\subsection{An approximate relaxation for CA-Node-Steiner-Tree}
\label{subsection:k-restricted}
From now on we focus on the CA-Node-Steiner-Tree problem. A crucial property that we will use is that, similar to the Edge Steiner Tree problem, one can show that the \emph{k-restricted} version of the problem provides a $(1+\varepsilon)$-approximate solution to the original problem, for an arbitrarily small $\varepsilon >0$. 

To explain this in more detail, we first recall that any Steiner tree can be seen as the union of components, where a component is a subtree whose leaves are all terminals, and whose internal nodes are all Steiner nodes. We note that two components of such a union are allowed to share nodes and edges. A component is called $k$-restricted if it has at most $k$ terminals. A $k$-restricted Steiner tree is a collection of $k$-restricted components whose union gives a feasible Steiner tree. The key theorem of this section is the following.

\begin{theorem}\label{thm:k-restricted-decomp}
Consider a CA-Node-Steiner-Tree instance and let $\OPT$ be its optimal value. For any integer $m \geq 1$, there exists a $k$-restricted Steiner tree $Q(k)$, for $k = 2^m$, whose cost satisfies
\begin{equation*}
    cost(Q(k)) \leq \left(1 + \frac{4}{\log k}\right) \OPT.
\end{equation*}
\end{theorem}
We stress here that $cost(Q(k))$ is equal to the number of Steiner nodes in $Q(k)$ counted with multiplicities; an example demonstrating this is given in Figure~\ref{fig:k-restricted-example}.

\input{./fig6.tex}

In proving the above theorem, we crucially use all the properties stated in Definition~\ref{def:ca-node-steiner}, as the statement is not true for general Node Steiner Tree instances. The proof of the above theorem is inspired by the result of Borchers and Du~\cite{DBLP:journals/siamcomp/BorchersD97} regarding the Edge Steiner Tree problem, and can be found in Appendix~\ref{appendix:k-restricted-decomp}.

The above theorem shows that CA-Node-Steiner-Tree can be approximated using $k$-restricted Steiner trees, with a small loss in the objective value. Based on this, we now present a linear programming relaxation for the $k$-restricted Node Steiner Tree problem, which we call the Directed Components Relaxation (or, in short, DCR), as it mimics the DCR relaxation of Byrka et al.~\cite{DBLP:journals/jacm/ByrkaGRS13} for the Edge Steiner Tree problem.

Let $G = (V, E)$ be a Node Steiner Tree instance, where $R \subseteq V$ is the set of terminals. Let $k > 0$ be an integer parameter. For each subset $R' \subseteq R$ of terminals of size at most $k$, and for every $c' \in R'$, we define a directed component $C'$ as the minimum Node Steiner tree on terminals $R'$ with edges directed towards $c'$. More precisely, let $S(R') \subseteq V \setminus R$ be a minimum-cardinality set of Steiner nodes such that $G[R' \cup S(R')]$ is connected. We define $C'$ by taking a spanning tree in $G[R' \cup S(R')]$ and directing the edges towards $c'$, and we let $cost(C') = |S(R')|$. We call $c'$ the \emph{sink} of the component and all other terminals in $R' \setminus \{c'\}$ the \emph{sources} of the component.
Let $\mathbf{C}$ be the set of all such directed components $C'$, with all possible sinks $c'\in C'$. Note that $\mathbf{C}$ has an element for each subset of terminals $R' \subseteq R$, $|R'| \leq k$, and each terminal $c' \in R'$. We introduce one variable $x(C')$ for each $C' \in \mathbf{C}$. Finally, we choose an arbitrary terminal $r \in R$ as the root. The $k$-DCR relaxation is the following.
\begin{align*}
\tag{$k$-DCR}
    \min: &\quad\sum_{C' \in \mathbf{C}} cost(C') \cdot x(C')\\
    \textrm{s.t.:}&\quad \sum_{\substack{C' \in \mathbf{C}:  \textrm{ sink}(C') \notin U \\\textrm{ and sources}(C') \cap U \neq \emptyset } } x(C') \geq 1, \quad\quad \forall U \subseteq R \setminus \{r\}, U \neq \emptyset,\\
                &\hspace{93.5pt} x(C') \geq 0, \quad\quad\forall C' \in \mathbf{C}.
\end{align*}

Let $\OPT_{\LP}(k)$ be the optimal value of the above relaxation for $k \in \mathbb{N}_{>0}$, 
and $\OPT$ be the value of an optimal (integral) solution to our original (unrestricted) CA-Node-Steiner-Tree instance. The next theorem states that the $k$-DCR LP can be solved in polynomial time for any fixed $k$ and its optimal solution yields a $(1+\varepsilon)$-approximation of $\OPT$. Its proof can be found in Appendix~\ref{appendix:proof-eps_approx}.
\begin{theorem}
\label{thm:eps_approx}
For any fixed $\varepsilon > 0$, there exists a $k = k(\varepsilon) > 0$ such that $\OPT_{\LP}(k) \leq (1+\varepsilon) \OPT$, and moreover, an optimal solution to the $k$-DCR LP can be computed in polynomial time.
\end{theorem}

Given the above theorem, all that remains to show is the following: given an optimal fractional solution to the $k$-DCR LP, design a rounding procedure that returns an integral CA-Node-Steiner-Tree solution whose cost is at most $\gamma$ times larger than the cost of an optimal fractional solution, for as small $\gamma \geq 1$ as possible. Following the chain of reductions discussed in this section, such a rounding scheme would immediately imply a $(\gamma+\varepsilon)$-approximation algorithm for 1-Node-CAP and CacAP, given that $\gamma$ is constant. In the following sections, we present such a scheme that gives $\gamma \leq 1.8917$.

%% file: fig6.tex
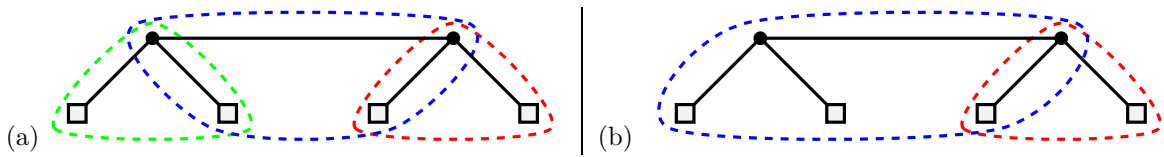
\begin{figure}[ht]
\centering
\begin{tabular}{c|c}
     (a)\begin{tikzpicture}
            
        
        \draw[very thick,color=green,dashed]
            plot [smooth cycle] coordinates {(0,0.2) (-1.3,-1.2) (1.3,-1.2)};
        \draw[very thick,color=red,dashed  ]            
            plot [smooth cycle] coordinates {(4,0.2) (2.7,-1.2) (5.3,-1.2)};
        \draw[very thick,color=blue,dashed]            
            plot [smooth cycle] coordinates {(-0.2,0.2) (4.2,0.2) (3.2,-1.2) (0.8,-1.2)};
            
        \filldraw[very thick]
            (0,0) circle(2pt) -- (4,0) circle(2pt)
            
            (0,0) -- (-1,-1) node[draw=black, fill=black!5]{}
            (0,0) -- (1,-1) node[draw=black, fill=black!5]{}
            
            (4,0) -- (3,-1) node[draw=black, fill=black!5]{}
            (4,0) -- (5,-1) node[draw=black, fill=black!5]{};

      \end{tikzpicture}
      &  (b)\begin{tikzpicture}
            
        
        \draw[very thick,color=red,dashed]
            plot [smooth cycle] coordinates {(4,0.2) (2.7,-1.2) (5.3,-1.2)};
            
       \draw[very thick,color=blue,dashed]
            plot [smooth cycle] coordinates {(-0.2,0.2) (4.2,0.2) (3.2,-1.2) (-1.2,-1.2)};
            
        \filldraw[very thick]
            (0,0) circle(2pt) -- (4,0) circle(2pt)
            
            (0,0) -- (-1,-1) node[draw=black, fill=black!5]{}
            (0,0) -- (1,-1) node[draw=black, fill=black!5]{}
            
            (4,0) -- (3,-1) node[draw=black, fill=black!5]{}
            (4,0) -- (5,-1) node[draw=black, fill=black!5]{};

      \end{tikzpicture}
\end{tabular}
      
        \endpgfgraphicnamed
        
        \caption{(a) An example of a 2-restricted CA-Node-Steiner tree, of cost $4$. (b) An example of a 3-restricted CA-Node-Steiner tree, of cost $3$. In both images, the square nodes represent the terminals and the components are the nodes grouped together by the dashed lines.}
        
\label{fig:k-restricted-example}
\end{figure}

%% file: alg.tex
\section{The iterative randomized rounding algorithm}

Our rounding scheme for the $k$-DCR relaxation consists in applying the iterative randomized rounding technique of Byrka et al.~\cite{DBLP:journals/jacm/ByrkaGRS13}, first applied to the Edge Steiner Tree problem, and is described in Algorithm~\ref{alg:rounding-scheme}.

\begin{algorithm}
\For{$i = 1, 2, 3,\dots$}
    {
    Compute an optimal solution $x^i$ to the $k$-DCR LP w.r.t.~the current instance;\\[3pt]
    Sample a component $C^i$ from the set $\mathbf{C}^i$ of current components, where $C^i =C$ with probability $x^i_C/\sum_{C'\in\mathbf{C}^i}x^i_{C'}$, and contract $C^i$ into its sink node;\\[3pt]    
    If a single terminal remains, return the sampled components $\bigcup^i_{j=1} C^j$;
    }
\caption{The iterative randomized rounding scheme for the $k$-DCR LP}
\label{alg:rounding-scheme}
\end{algorithm}

We will now analyze the cost of the output solution. Throughout this section, we follow the analysis of \cite{DBLP:journals/jacm/ByrkaGRS13} and slightly modify it wherever needed.

Let $G = (V, E)$ be a given CA-Node-Steiner-Tree instance, where $V = R \cup S$; $R$ is the set of terminals and $S$ is the set of Steiner nodes. Let $\varepsilon > 0$ be a fixed small constant and $k = k(\varepsilon)$ as in the proof of Theorem~\ref{thm:eps_approx}. For the sake of the analysis, we can assume that $\sum_{C'\in\mathbf{C}^i}x^i_{C'}$ is the same for all iterations $i$, as in \cite{DBLP:journals/jacm/ByrkaGRS13}\footnote{Roughly speaking, in order to ensure that the sum is the same during all iterations, we can slightly modify the algorithm and add a dummy component of zero cost whose $x$-value is set in such a way so that the summation is equal to some $M > 0$. A coupling argument shows that all expected properties of the output are the same for both versions of the algorithm. See [page 12, \cite{DBLP:journals/jacm/ByrkaGRS13}] for more details.}. So let $M =\sum_{C'\in\mathbf{C}^i}x^i_{C'}$.

Let $T = (R \cup S^*, E^*)$, $S^* \subseteq S$ and $E^* \subseteq E$, be an optimal solution to the given CA-Node-Steiner-Tree instance $G$ of cost $cost(T) = |S^*|$. We will analyze the cost of the output of Algorithm~\ref{alg:rounding-scheme} with respect to $T$. For that, we define a sequence of subgraphs of $T$, one for each iteration of the algorithm, in the following way: if, during the $i^\textrm{th}$ iteration, we sample $C^i$, then we delete a subset of nodes of $T$ and what remains is the subgraph $T^i$ defined for that iteration. What must be specified is which nodes of $T$ are deleted at each iteration, which will be explained shortly. Let $T=T^0  \sqsupseteq  T^1  \sqsupseteq  T^2  \sqsupseteq  ...$ be sequence of subgraphs remaining after each iteration, where the notation $T^{i+1} \sqsubseteq T^i$ means that $T^{i+1}$ is a (not necessarily strict) subgraph of $T^i$. We show that there exists some universal constant $\gamma \in \mathbb{R}_{\geq 1}$ and a choice of optimal solution $T$ such that the following two properties are satisfied:
\begin{enumerate}
    \item[(a)] $T^i$ plus the components sampled until iteration $i$ form a connected subgraph spanning the terminals.
    \item[(b)] On average, a Steiner node in $T$ is deleted after $M \gamma $ iterations.
\end{enumerate}

Similar to \cite{DBLP:journals/jacm/ByrkaGRS13}, we can show that conditions (a) and (b) yield that the iterative randomized rounding algorithm is in fact a $(\gamma + \varepsilon)$-approximation algorithm, for any fixed $\varepsilon > 0$. This is achieved by relying on the construction of a \emph{witness tree} $W$, which is a particular kind of spanning tree on the set of terminals. However, the main differences with respect to \cite{DBLP:journals/jacm/ByrkaGRS13} are that (i) we delete \emph{nodes} of $T$ instead of edges, and (ii) we have a \emph{purely deterministic} way to construct the witness tree, and thus we need an explicit averaging argument for (b). 

We now discuss the details of our deletion process. As mentioned,  given $T$ we construct a witness tree $W = (R, E_W)$ that spans the set of terminals. For each component $C$, let $\mathcal B_W(C)$ be the family of maximal edge sets $B \subseteq E_W$ such that $(W\setminus B) \cup C$ forms a connected subgraph spanning the terminals. In each iteration $i$, we ``mark'' a subset of edges of $W$ that correspond to a randomly chosen set in $\mathcal B_W(C^i)$. For a positive integer $t$, let $H(t) \coloneqq \sum_{j=1}^t \frac{1}{j}$ be the $t^\textrm{th}$ harmonic number. The following lemma is proved in~\cite{DBLP:journals/jacm/ByrkaGRS13} (more precisely, see Lemmas 19 and 20 in~\cite{DBLP:journals/jacm/ByrkaGRS13}).

\begin{lemma}[\cite{DBLP:journals/jacm/ByrkaGRS13}]
\label{lem:probability_distr}
For each component $C$, there exists a probability distribution over $\mathcal B_W(C)$ such that the following holds: 
for any $\widetilde W \subseteq E_W$, the expected number of iterations until all edges of $\widetilde W$ are marked is bounded by $H(|\widetilde W|) \cdot M$, where the expectation is over the random choices of the algorithm and the distributions over $\{\mathcal B_W(C)\}_C$.
\end{lemma}

We will delete Steiner nodes from $T$ using the marked edges in $W$, as follows. For each Steiner node $v \in S^*$, we define
\begin{equation*}
    W(v) \coloneqq \{(p,q) \in E_W: v \textrm{ is an internal node of the path between }p \textrm{ and }q \textrm{ in } T\},
\end{equation*}
and $w(v) \coloneqq |W(v)|$; less formally, $w(v)$ is equal to the number of edges $e$ of $W$ such that the path between the endpoints of $e$ in $T$ contains $v$. From now on, we will say that the vector $w: S^* \to \mathbb{N}_{\geq 0}$ is the vector \emph{imposed} on $S^*$ by $W$. A Steiner node $v \in S^*$ is deleted in the first iteration where all the edges in $W(v)$ become marked. Thus, each subgraph $T^i$, $i > 0$, defined in the beginning of this section, is the subgraph obtained from $T^{i-1}$ after the $i^\textrm{th}$ iteration according to this deletion process.

\begin{lemma}
\label{lem:connectivity}
For every $i > 0$, $T^i \cup C^1 \cup \dots \cup C^i$ spans the terminals.
\end{lemma}

\begin{proof}
Let $E_W^i \subseteq E_W$ be the set of edges of $W$ that have not been marked at iteration $i$. By construction, 
$E_W^i \cup C^1 \cup \dots \cup C^i$ spans the terminals. For every $(p,q) \in E_W^i$, the unique path between $p$ and $q$ in $T$ is still present in $T^i$, as none of the nodes in this path have been deleted. Hence $p$ and $q$ are connected in $T^i$. The result follows.
\end{proof}

We now address the universal constant $\gamma$ that was mentioned above. For a CA-Node-Steiner-Tree instance $G=(V, E)$, where $R \subseteq V$ is the set of terminals and $S = V \setminus R$ is the set of Steiner nodes, we define
\begin{equation*}
    \gamma_G \coloneqq \min_{\substack{T = (R \cup S^*, E^*):\; T \textrm{ is}\\\textrm{optimal Steiner tree}}} \min_{\substack{W:\; W\textrm{ is }\\\textrm{witness tree}}} \frac{\sum_{v \in S^*} H(w(v))}{cost(T)},
\end{equation*}
where $w$ is the vector imposed on $S^*$ by each witness tree $W$ considered. The constant $\gamma$ that will be used to pinpoint the approximation ratio of the algorithm is defined as
\begin{equation*}
    \gamma \coloneqq \sup \{\gamma_G: G\textrm{ is an instance of CA-Node-Steiner-Tree}\}.
\end{equation*}

In Section~\ref{sec:witness-tree} we will prove the following theorem.

\begin{theorem}\label{thm:witness-tree}
$\gamma \leq 1.8917$.
\end{theorem}

We are ready to prove the main theorem of this section; throughout its proof, we use the notation introduced in this section.



\begin{theorem}\label{thm:main-approx}
For any fixed $\varepsilon > 0$, the iterative randomized rounding algorithm (see Algorithm~\ref{alg:rounding-scheme}) yields a $(\gamma + \varepsilon)$-approximation.
\end{theorem}

\begin{proof}
Let $G = (R \cup S, E)$ be a CA-Node-Steiner-Tree instance. We run Algorithm~\ref{alg:rounding-scheme} with $k = 2^{\lceil 4/\varepsilon' \rceil}$, as in the proof of Theorem~\ref{thm:eps_approx}, and $\varepsilon' \coloneqq   \frac{\varepsilon}{2}$. It is easy to see that the algorithm runs in polynomial time.

Let $T = (R \cup S^*)$ and $W$ be the optimal Steiner tree and corresponding witness tree, respectively, such that $\gamma_G = \frac{\sum_{v \in S^*} H(w(v))}{cost(T)}$. Clearly, $\gamma_G \leq \gamma$. For each Steiner node $v \in S^*$, let $D(v) \coloneqq \max\{i: v\in T^i\}$. By Lemma~\ref{lem:probability_distr}, we have $\E[D(v)] \leq H(w(v)) \cdot M$. 
We have
\begin{align*}
\sum_{i}\E[cost(C^i)] &=\sum_{i}\E\Big[\sum_{C} \frac{x^i_C}{M} cost(C)\Big] \leq \frac{1+\varepsilon'}{M} \sum_{i} \E[cost(T^i)]= \frac{1+\varepsilon'}{M}\sum_{v \in S^*} \E[D(v)]\\
                     &\leq (1+\varepsilon') \sum_{v \in S^*} H(w(v)) = \gamma_G \left(1 +\frac{\varepsilon}{2} \right) cost(T) \leq (\gamma +\varepsilon) cost(T).
\end{align*}

The first inequality holds since (i) Lemma~\ref{lem:connectivity} implies that an optimal solution to the CA-Node-Steiner-Tree instance defined in iteration $i$ has cost at most $cost(T^i)$, and (ii) Theorem~\ref{thm:eps_approx} shows that the $k$-DCR LP provides a $(1+\varepsilon')$-approximate solution to it. The first inequality of the second line follows since $\sum_{v \in S^*} \E[D(v)] \leq M \cdot \sum_{v \in S^*} H(w(v))$, while the last inequality follows by the definition of $\gamma$ and Theorem~\ref{thm:witness-tree}.
\end{proof}

Putting everything together, we can now easily prove our main two theorems, Theorems~\ref{thm:approx-block-tap} and~\ref{thm:approx-cacap}.
\begin{proof}[Proof of Theorems~\ref{thm:approx-block-tap} and~\ref{thm:approx-cacap}]
We first design a $1.892$-approximation algorithm for the CA-Node-Steiner-Tree problem. For that, we pick $\varepsilon$ to be sufficiently small (e.g., $\varepsilon = 0.0003$ suffices) and by Theorems~\ref{thm:witness-tree} and \ref{thm:main-approx}, we get that the iterative randomized rounding algorithm (Algorithm~\ref{alg:rounding-scheme}) is a $1.892$-approximation algorithm for CA-Node-Steiner-Tree. By Theorems~\ref{cor:approx-preserve} and~\ref{cor:approx-preserve-block}, we also get a $1.892$-approximation algorithm for 1-Node-CAP and CacAP. This concludes the proof.
\end{proof}

%% file: saving_tree.tex
\section{A deterministic construction of the witness tree}
\label{sec:witness-tree}

In this section, for any given CA-Node-Steiner-Tree instance $G=(V, E)$, we consider a specific optimal Steiner tree and describe a deterministic construction of a witness tree $W$ that shows that $\gamma_G < 1.8917$. By the definition of $\gamma$, this immediately implies Theorem~\ref{thm:witness-tree}. 

Throughout this section, we use the following notation: given is a CA-Node-Steiner-Tree instance $G=(R \cup S, E)$ and an optimal solution $T = (R \cup S^*, E^*)$. The goal is to construct a witness tree $W = (R, E_W)$ spanning the terminals, such that $\E[H(w)] \coloneqq \frac{1}{|S^*|} \sum_{v \in S^*} H(w(v))$ is minimized, where $w: S^* \to \mathbb{N}_{\geq 0}$ is the vector imposed by $W$. From now on, we will refer to $\E[H(w)]$ as the average $H(w)$-value of the Steiner nodes. Similarly, we will refer to $w(v)$ as the $w$-value of a node $v$. Since the proof of Theorem~\ref{thm:witness-tree} consists of several steps, we will present them separately, and at the end of the section we will put everything together and formally prove Theorem~\ref{thm:witness-tree}.

We start by observing that every CA-Node-Steiner-Tree instance has an optimal solution in which the terminals are leaves. This follows by Definition~\ref{def:ca-node-steiner}, since (i) there are no edges between terminals and (ii) adjacent Steiner nodes of a terminal form a clique. Thus, from now on we assume that $T$ is an optimal solution in which the terminals are leaves. Given this, the first step is to show that we can reduce the problem of constructing a witness tree spanning the terminals to the problem of constructing a witness tree spanning the Steiner nodes that are adjacent to terminals. This will allow us to remove the terminals from $T$, and only focus on $T[S^*]$, which is a (connected) tree.

\subsection{Removing terminals from $T$}\label{sec:removing-terminals}

Let $F \subseteq S^*$ be the set of Steiner nodes contained in $T$ with at least one adjacent terminal. We call $F$ the set of \emph{final} Steiner nodes. In order to simplify notation, we will remove the terminals from $T$, and compute a tree $W$ spanning the final nodes of $T[S^*]$ (note that the set of final nodes contains all the leaves of $T[S^*]$). Any such tree $W$ directly corresponds to a terminal spanning tree as follows: for each final Steiner node $f \in F$, we arbitrarily pick one of the (at most) two terminals adjacent to $f$, which we denote as $\mathtt{rep}(f)$. Given a spanning tree $W$ of the final Steiner nodes, we can now generate a terminal spanning tree $W'$ as follows (see Figure~\ref{fig:Removing-terminals}a):
\begin{itemize}
    \item we replace every edge $(f, f')$ of $W$ with the edge $(\mathtt{rep}(f), \mathtt{rep}(f'))$,
    \item in case a final Steiner node $f$ has a second terminal $t$ adjacent to it besides $\mathtt{rep}(f)$, we also add the edge $(\mathtt{rep}(f), t)$ to the tree.  
\end{itemize}

It is easy to see that the vectors $w$ and $w'$ imposed by $W$ and $W'$, respectively, on the Steiner nodes are the same for all nodes except for the final Steiner nodes, where the $w$-values differ by $1$ in the case where a final Steiner node has two terminals adjacent to it. Thus, from now on, we focus on computing a tree $W$ spanning the final Steiner nodes, and we simply increase the corresponding imposed $w$-value of each final Steiner node by $1$. More formally, if $W = (S^*, E_W)$ is a tree spanning the final nodes of $T[S^*]$, we have for every $v \in S^*$:
\begin{equation*}
    w(v) \coloneqq |\{(p,q) \in E_W: v \textrm{ is a node of the path between }p \textrm{ and }q \textrm{ in } T[S^*]\}| + \mathbbm{1}[v \in F],
\end{equation*}
where $\mathbbm{1}[v \in F]$ denotes the indicator of the event ``$`v \in F$". This expression corresponds to the (worst-case) assumption that every final Steiner node is incident to two terminals; this is without loss of generality since, if the average $H(w)$-value computed in this way is less than a given value $\gamma' > 0$, then clearly the average $H(w')$-value imposed by the terminal spanning tree $W'$ is also bounded by $\gamma'$, as the harmonic function $H$ is an increasing function.

The next step is to root $T[S^*]$ and then decompose it into a collection of rooted subtrees of $T[S^*]$ such that every final Steiner node appears as either a leaf or a root (or both); we call these subtrees \emph{final-components}. This decomposition is nearly identical to the one used in Section \ref{subsection:k-restricted} where we decomposed a tree into components where terminals appeared as leaves (here we use final Steiner nodes instead of terminals), with the added element that now these components are also rooted.

\input{./fig7.tex}

\subsection{Decomposing $T$ into final-components}\label{sec:decomposition}

To simplify notation, let $T = (V, E)$ be a tree where $F \subseteq V$ is the set of final Steiner nodes, again noting that $F$ contains all the leaves of $T$. We start by rooting the tree $T$ at an arbitrary leaf $r$; note that $r \in F$. We decompose $T$ into a set of final-components $T_1, \dots, T_{\tau}$, where a final-component is a rooted maximal subtree of $T$ whose root and leaves are the only nodes in $F$. We clarify here that for a rooted tree, we call a vertex a leaf if it has no children (thus, although the root may have degree $1$, it is not referred to as a leaf).

More precisely, we start with $T$ and consider the maximal subtree $T_1$ of $T$ that contains the root $r$, such that all leaves of $T_1$ are final Steiner nodes, and including the root, are the only final Steiner nodes. For each leaf $f \in F \cap T_1$, we do the following. If $f$ has children $\{x_1, \ldots, x_p\}$ in $T$, for some $p \in \mathbb{N}_{p \geq 1}$, we first create $p$ copies of $f$; this means that, along with the original node $f$, we now have $p+1$ copies of $f$, and one copy of it remains in $T_1$. We then remove $T_1$ from $T$, and for each tree $T(x_i)$, $i \in [p]$, rooted at $x_i$, that is part of the forest $T \setminus T_1$, we connect $x_i$ with one distinct copy of $f$, and set this copy to be the root of the modified tree $T(x_i)$. Doing this for every leaf of $T_1$, we end up with a forest. We repeat this procedure recursively for each tree of this forest, until all trees of the forest are final-components. A simple example of this decomposition is shown in Figure~\ref{fig:Removing-terminals}b. Let $T_1, \ldots, T_\tau$ be the resulting final-components. Without loss of generality, we assume that for every $i \in [\tau]$, $\bigcup_{j\leq i} T_j$ is connected. We clarify that in this union of trees, duplicate copies of the same node are ``merged" back into one node.

Given this collection $T_1, \ldots, T_\tau$ of final-components, we will now show how to compute a tree $W_i$ spanning the final Steiner nodes of each $T_i$, and then show how to combine all of them to generate a tree $W$ that spans the final Steiner nodes of $T$, while bounding the average $H(w)$-value imposed on $S^*$ by $W$.

\subsection{Computing a tree $W$ spanning the final Steiner nodes of $T$}\label{sec:main-analysis}

We start by describing how to compute a tree $W_i$ that spans the final Steiner nodes of each final-component $T_i$, as well as how to join these trees $\{W_i\}_{i = 1}^\tau$ to obtain a tree $W$ spanning the whole set of final Steiner nodes of $T$. We construct $W$ iteratively as follows. We start with $W =\emptyset$, process the final-components $T_1, \ldots, T_\tau$ one by one according to the index-order, and for each $i$, compute a tree $W_i$ for $T_i$ and merge it with the current $W$; again, in this merging, multiple copies of the same node are merged back into one node.

To analyze this procedure, we fix a final-component $T_i$ rooted at $r_i$, and let $(r_i, v)$ be the unique edge incident to $r_i$ inside $T_i$. Let $T' = \bigcup_{j < i} T_j$, and $W'$ be the constructed witness tree for $T'$. Let $w'$ be the vector imposed on the Steiner nodes of $T'$ by $W'$. We do some case analysis based on whether $v \in F$ or not; throughout this analysis, we use the notation $|T|$ to denote the number of nodes of a tree $T$.

\bigskip
\noindent \textbf{Case 1: $v \in F$.} In this case, the edge $(r_i,v)$ is the only edge of $T_i$, so $W_i$ simply consists of the nodes $\{r_i, v\}$ and the edge $(r_i,v)$. The following lemma shows that merging $W_i$ with $W'$ (by simply merging the two copies of $r_i$ and taking the union of the edges of $W'$ and $W_i$, in the case where $W' \neq \emptyset$) increases the average $H(w)$-value by at most a negligible amount. Let $w''$ be the vector imposed by $W'' = W' \cup W_i$. 

\begin{lemma}\label{lem:h3}
Let $\gamma' \geq H(3)$. If $\sum_{u \in T'} \frac{H(w'(u))}{|T'|} < \gamma'$, then
$\sum_{u \in T'\cup T_i} \frac{H(w''(u))}{|T' \cup T_i|} < \gamma'$.
\end{lemma}

\begin{proof}
We do some case analysis. If $i = 1$, then $W' = \emptyset$ and $r_i = r$. In this case, we have $w''(r_i) = w''(v) = 2$. This means that $\sum_{u \in T'\cup T_i} \frac{H(w''(u))}{|T' \cup T_i|} = H(2) < \gamma'$.

So, suppose now that $i > 1$, which implies that $W' \neq \emptyset$. Note that $w''(v) = 2$, $w''(r_i) = w'(r_i) + 1$, and for each $u \in T'\setminus \{r_i\}$, $w''(u)= w'(u)$. Moreover, we have $w'(r_i) \geq 2$; to see this, note that $r_i$ is a final Steiner node of $T'$ and has degree at least $1$ in $W'$. Thus, $H(w''(r_i)) \leq H(w'(r_i)) + \frac{1}{3}$. Therefore,
\begin{equation*}\sum_{u \in T'\cup T_i} \frac{H(w''(u))}{|T' \cup T_i|} \leq \frac{H(w''(v)) + 1/3  + \sum_{u \in T'} H(w'(u))}{|T' \cup T_i|} =
\frac{H(3)  + \sum_{u \in T'} H(w'(u))}{|T'| + 1} < \gamma'.
\end{equation*}
We conclude that the lemma always holds.
\end{proof}

\bigskip
\noindent \textbf{Case 2: $v \notin F$.} In this case, the witness tree $W_i$ of $T_i$ is computed deterministically by a procedure described in Algorithm~\ref{alg:computing_w}; we provide an example of the output of Algorithm~\ref{alg:computing_w} in Figure~\ref{fig:deterministic-witness-trees}.

\begin{algorithm}[ht]
Assign an arbitrary ordering to the leaves of $T_i$;\\[3pt]
Assign a weight to each node $u \in T_i$ that is equal to the inverse of its degree in $T_i$;\\[3pt]
For each non-final Steiner node $u \in T_i$, compute a minimum-weight path $P(u)$ from $u$ to a leaf that is a descendant of $u$ in $T_i$, where the weight of the path is given by the sum of weights of its nodes. Break ties in favor of the smallest leaf according to the ordering computed at step (1);\\[3pt]
For each non-final Steiner node $u \in T_i$, contract the path $P(u)$ computed in step (3) into its leaf endpoint, and obtain a tree spanning the final Steiner nodes of $T_i$; use this as the witness tree $W_i$;

\caption{Computing the tree $W_i$}
\label{alg:computing_w}
\end{algorithm}  

\input{./fig2.tex}

In order to bound the average $H(w)$-value that we get after adding $W_i$ to the current $W'$, we first introduce some notation. For any node $u \neq r_i$ of $T_i$, let $\mathtt{parent}(u)$ denote the parent of $u$ in $T_i$. Let $Q_u$ denote the subtree of $T_i$ rooted at $u$. For any non-final Steiner node $u$, $P(u)$ denotes the minimum-weight path from $u$ to a leaf, as computed at step (3) of Algorithm~\ref{alg:computing_w}. In the case where $u$ is a final Steiner node of $T_i$, we define $P(u) \coloneqq \{u\}$; in particular $P(r_i) \coloneqq \{r_i\}$. For every $u \in T_i$, let $q_u$ be the number of nodes of $P(u)$ (including the endpoints) and $l(u)$ denote the (unique) endpoint of this path that is a final Steiner node. For every non-final Steiner node $u$, exactly one child of $u$ is in this path, which we call the \emph{marked} child of $u$. For every $u \neq r_i$, we let $a(u)$ denote the node closest to $u$ on the path from $u$ to $r_i$ such that $l(u) \neq l(a(u))$. Examples of these definitions are given in Figure~\ref{fig:definitions}.

\input{./fig1.tex}

To begin the analysis, we first observe that for every node $u \neq r_i$ of $T_i$, $W_i$ contains an edge with endpoints $l(u)$ and $l(a(u))$; we call this edge $e_u$. Note that we can have $e_u = e_v$ for $u \neq v$. For every $u \neq r_i$, we now consider the subtree $W^u$ of $W_i$ on $Q_u$, defined as a set of edges as follows:
\begin{equation*}
W^u = \{e \in W_i :  \mbox{ both the endpoints of $e$ are in $Q_u$} \} \cup e_u.
\end{equation*}
We also define $W^{r_i} = W_i$. Observe that for every $u \in T_i$, $W^u$ indeed corresponds to a connected subtree of $W_i$. An example of $W^u$ is provided in Figure~\ref{fig:wu-description}, where we also clarify several properties of Lemma~\ref{lem:increase}, proved below. We let $w^u$ be the vector imposed by $W^u$, where for every $\ell \in Q_u$, we have
\begin{equation*}
    w^u(\ell) \coloneqq |\{(p,q) \in W^u: \ell \textrm{ is an internal node of the path between }p \textrm{ and }q \textrm{ in } T_i\}| + \mathbbm{1}[\ell \textrm{ is a leaf of }Q_u].
\end{equation*}
Let $h(Q_u) \coloneqq \sum_{\ell \in Q_u} H(w^u(\ell))$. We note here that $W^u$ might involve one final Steiner node, namely $l(a(u))$, that is not contained in $Q_u$. Finally, let $d_u$ be equal to $3$ if $u$ is a leaf, otherwise, let $d_u$ be equal to the degree of $u$ in $T_i$. 

\input{./fig5.tex}

\begin{lemma}\label{lem:increase}
Let $u \in T_i$  be a non-final Steiner node and $u_1, \dots, u_p$ be the children of $u$, with $u_1$ being the marked child of $u$. Then:
\begin{itemize} 
\item[(a)] $w^u(u) = p$.
\item[(b)] For every $j \in \{2, \ldots, p\}$ and every Steiner node $\ell \in Q_{u_j}$, $w^u(\ell) = w^{u_j}(\ell)$.
\item[(c)] For every $\ell \in Q_{u_1} \setminus P(u_1)$, $w^u(\ell) = w^{u_1}(\ell)$.
\item[(d)] For every $\ell \in P(u_1)$, $H(w^u(\ell)) - H(w^{u_1}(\ell)) \leq \sum_{j=2}^{p} \frac{1}{d_{\ell} +j-2}$.
\end{itemize}
\end{lemma}

\begin{proof}
    We first observe that $W^u$ is given by the disjoint union of $W^{u_1} \cup W^{u_2} \cup \ldots \cup W^{u_p}$. In fact, for $j=2, \ldots, p$, the edge $e_{u_j} \in W^{u_j}$ is precisely the edge with endpoints $l(u_1)$ and $l(u_j)$, and the edge $e_{u_1} \in W^{u_1}$ corresponds to $e_u$. Given these observations, $ w^u(u) = p$ as the edges of $W^u$ which contribute to $w^u(u)$ are precisely $e_{u_1}, \dots, e_{u_p}$. Hence (a) holds. 
    
    Statements (b) and (c) follow by observing that the only nodes whose $w^u$-value differs from the corresponding $w^{u_j}$-value are in $Q_{u_1}$, and in particular they are precisely the nodes of the path $P(u_1)$. It is not hard to see that for every $\ell \in P(u_1)$, we have $w^u(\ell) = w^{u_1}(\ell) + p-1$. Thus, we get that
    \begin{equation*}
        H(w^u(\ell)) - H(w^{u_1}(\ell)) = \sum_{j= 2}^{p} \frac{1}{w^{u_1}(\ell) + j - 1}.
    \end{equation*}
    We will now compute a lower bound for $w^{u_1}(\ell)$. If $\ell \neq l(u_1)$, then we can apply statement (a), which we just proved, to $Q_{u_\ell}$, and get that $d_{\ell} -1 = w^{\ell}(\ell) \leq w^{u_1}(\ell)$. Hence (d) holds in this case. We now consider the case $\ell = l(u_1)$, i.e., $\ell$ is a leaf of $Q_{u_1}$. In this case, we have $d_{\ell} =3$ and $w^{\ell}(\ell) =2$, where the latter equality holds since there is a unique edge $e_{\ell} \in W^{\ell}$ and $\ell$ is a leaf, which means that its $w$-value is increased by $1$. Thus, statement (d) holds in this case as well. This concludes the proof.
\end{proof}

The following lemma is the main technical part of our analysis.

\begin{lemma}\label{lem:invariant}
For each node $u \in T_i \setminus r_i$, and $\delta = 7/120$, we have
\begin{equation*}
h(Q_u) + \sum_{\ell \in P(u)} \frac{1}{d_{\ell}} + \delta < 1.8917 \cdot |Q_u|.
\end{equation*}
\end{lemma}

\begin{proof}
The proof is by induction on $|Q_u|$. The base case is $|Q_u| =1$, in which case $u$ is leaf of $T_i$. In this case, we have $d_u =3$, $w^u(u)=2$, and $P(u)=\{u\}$, giving
\begin{equation*}
h(Q_u) + \frac{1}{3}  + \delta  = H(2)  + \frac{1}{3} + \delta = 1.891\bar{6} < 1.8917.
\end{equation*}
    
For the induction step, let $u$ be a non-final Steiner node of $T_i$, let $u_1, \dots, u_p$ be its children, and let $u_1$ be its marked child. By Lemma~\ref{lem:increase} we have that
\begin{equation*}
h(Q_u) = \sum_{j=1}^p h(Q_{u_i}) + \sum_{\ell \in P(u_1)} \left(H(w^{u}(\ell)) - H(w^{u_1}(\ell))\right)+ H(p). 
\end{equation*}
By the induction hypothesis, for each $j \in [p]$ we have $h(Q_{u_j}) + \sum_{\ell \in P(u_j)} \frac{1}{d_{\ell}} + \delta < 1.8917 \cdot |Q_{u_j}|$. Combining these we get
\begin{align*}
    h(Q_u) &< 1.8917 \sum_{j=1}^p|Q_{u_j}| - \delta p - \sum_{j=1}^p \sum_{\ell \in P(u_j)} \frac{1}{d_{\ell}} + \sum_{\ell \in P(u_1)} (H(w^{u}(\ell)) - H(w^{u_1}(\ell)))+ H(p)\\
    &\leq 1.8917 \sum_{j=1}^p|Q_{u_j}| - \delta p - \sum_{j=1}^p \sum_{\ell \in P(u_j)} \frac{1}{d_{\ell}} + \sum_{\ell \in P(u_1)} \sum_{j = 2}^p \frac{1}{d_\ell + j - 2} + H(p),
\end{align*}
where the last inequality follows due to statement (d) of Lemma~\ref{lem:increase}. We now observe that, due to the greedy choice of the path $P(u)$ at step (3) of Algorithm \ref{alg:computing_w}, for each $j \in [p]$ we have
\begin{equation*}
    \sum_{\ell \in P(u_j) \setminus F} \frac{1}{d_{\ell}} + 1 \geq \sum_{\ell \in P(u_1) \setminus F} \frac{1}{d_{\ell}} + 1,
\end{equation*}
since there is exactly one final Steiner node in each path $P(u_j)$, $j \in [p]$. Thus, we conclude that $\sum_{\ell \in P(u_j)} \frac{1}{d_{\ell}} \geq \sum_{\ell \in P(u_1)} \frac{1}{d_{\ell}}$ for every $j \in [p]$. Therefore, $\sum_{j=1}^p \sum_{\ell \in P(u_j)} \frac{1}{d_{\ell}} \geq p \sum_{\ell \in P(u_1)} \frac{1}{d_{\ell}}$. Moreover, for every $\ell \in P(u_1)$, we have $\sum_{j = 2}^p \frac{1}{d_\ell + j - 2} \leq \frac{p-1}{d_\ell}$. Combining these we get
\begin{align*}
    h(Q_u) &< 1.8917 \sum_{j=1}^p|Q_{u_j}| - \delta p  - p \sum_{\ell \in P(u_1)} \frac{1}{d_{\ell}} + (p-1)\sum_{\ell \in P(u_1) \setminus \{l(u)\}} \frac{1}{d_\ell} + \sum_{j = 2}^p \frac{1}{d_{l(u)} + j - 2} + H(p)\\
    &\leq 1.8917 \sum_{j=1}^p|Q_{u_j}| - \delta p - \sum_{\ell \in P(u_1)} \frac{1}{d_{\ell}} - \frac{p-1}{d_{l(u)}} + \sum_{j = 2}^p \frac{1}{d_{l(u)} + j - 2} + H(p).
\end{align*}
Substituting $d_{l(u)}=3$ into this equation, we get
\begin{equation*}
    h(Q_u) < 1.8917 \sum_{j=1}^p|Q_{u_j}| - \delta p - \sum_{\ell \in P(u_1)} \frac{1}{d_{\ell}} - \frac{p-1}{3} + \sum_{j=2}^{p} \frac{1}{j+1} + H(p).
\end{equation*}
Rearranging, and using  $\sum_{\ell \in P(u)} \frac{1}{d_{\ell}} = \sum_{\ell \in P(u_1)} \frac{1}{d_{\ell}} + \frac{1}{d_{u}}$ and $d_u = p+1$, we get
\begin{align*}
    h(Q_u) + \sum_{\ell \in P(u)} \frac{1}{d_{\ell}} + \delta &< 1.8917 \sum_{j=1}^p|Q_{u_j}| - \delta (p-1) - \frac{p-1}{3} + \sum_{j=2}^{p} \frac{1}{j+1} + H(p) + \frac{1}{d_u}\\
    &= 1.8917 \sum_{j=1}^p|Q_{u_j}| - \delta (p-1) - \frac{p-1}{3} + \sum_{j=2}^{p} \frac{1}{j+1} + H(p+1).
\end{align*}
We will now show that
\begin{equation*}
    -\delta (p-1) - \frac{p-1}{3} + \sum_{j=2}^{p} \frac{1}{j+1} + H(p+1) \leq 1.8917.
\end{equation*}
For that, we define
\begin{equation*}
    \psi(p) \coloneqq 2H(p+1) - (p-1)\left(\frac{1}{3} + \delta \right) - H(2).
\end{equation*}
Note that the left-hand side of the inequality above is equal to $\psi(p)$. We now show that $\psi(p) \leq \psi(4)$ for every $p \in \mathbb{N}_{>0}$. For that, we first consider $p \geq 4$, and we have
\begin{align*}
    \psi(p+1) - \psi(p) &= 2H(p+2) - p\left(\frac{1}{3} + \delta \right) - H(2) - \left(2H(p+1) - (p-1)\left(\frac{1}{3} + \delta \right) - H(2) \right)\\
            &= \frac{2}{p+2} - \frac{1}{3} - \delta \leq \frac{2}{6} - \frac{1}{3} - \delta = -\delta < 0.
\end{align*}
Thus, $\psi(p)$ is decreasing for $p \geq 4$. This means that $\max_{p \in \mathbb{N}_{>0}} \psi(p) = \max\{\psi(1), \psi(2), \psi(3), \psi(4)\}$. It is easy to verify that $\max_{p \in \mathbb{N}_{>0}} \psi(p) = \psi(4) = 1.891\bar{6} < 1.8917$. Putting everything together, we conclude that
\begin{equation*}
    h(Q_u) + \sum_{\ell \in P(u)} \frac{1}{d_{\ell}} + \delta < 1.8917 \cdot |Q_u|.
\end{equation*}
\end{proof}

Finally, we observe that $W_i = W^{r_i}$. With the above lemma at hand, it is not difficult to show that merging $W_i$ with the current $W$ does not increase the average $H(w)$-value by too much. Once again, let $T'$ be the union of $T_1, \dots, T_{i-1}$, let $w'$ be the vector imposed by the current $W'$ before adding $W_i$, and $w''$ be the vector imposed by $W' \cup W_i$.

\begin{lemma}\label{lem:finalbound}
If $\sum_{u \in T'} \frac{H(w'(u))}{|T'|} < 1.8917$, then $\sum_{u \in T'\cup T_i} \frac{H(w''(u))}{|T' \cup T_i|} < 1.8917$.
\end{lemma}
\begin{proof}
    We do some case analysis. We first consider the case of $i = 1$, which means that $r_i = r$ and $W' = \emptyset$. In this case, since $T_i \setminus \{r_i\} = Q_v$, we can apply Lemma~\ref{lem:invariant} and get that $\sum_{\ell \in T_i \setminus \{r_i\}} H(w''(\ell)) < 1.8917(|T_i|-1) - \sum_{\ell \in P(v)} \frac{1}{d_{\ell}} - \delta$. Moreover, we have $w''(r_i) = 2$. Thus, we get
    \begin{align*}
        \sum_{u \in T'\cup T_i} \frac{H(w''(u))}{|T' \cup T_i|} &= \frac{\sum_{\ell \in T_i \setminus \{r_i\}} H(w''(\ell)) + H(2)}{|T' \cup T_i|} < \frac{1.8917 (|T_i|-1)  + H(2)}{|T' \cup T_i|}  < 1.8917.
    \end{align*}

    We now turn to the case of $i > 1$, which means that $W' \neq \emptyset$. In this case, we note the following:
    \begin{enumerate}
        \item $w''(\ell) = w'(\ell)$ for all $\ell \in T_i \setminus r_i$,
        \item $w''(r_i) = w'(r_i) + 1$, and 
        \item $w'(u)=w''(u)$ for all $u \in T'\setminus \{r_i\}$.
    \end{enumerate}
    Since $r_i$ is a leaf of $T'$, we have $w'(r_i) \geq 2$, which implies that $H(w''(r_i)) \leq H(w'(r_i)) + \frac{1}{3}$. Moreover, since $T_i \setminus \{r_i\} = Q_v$, we can apply  Lemma~\ref{lem:invariant} to get $ \sum_{\ell \in T_i \setminus \{r_i\}} H(w''(\ell)) < 1.8917(|T_i|-1) - \sum_{\ell \in P(v)} \frac{1}{d_{\ell}} - \delta$. 
    Furthermore, $ \sum_{\ell \in P(v)} \frac{1}{d_{\ell}} + \delta \geq \frac{1}{3}$, as $d_{l(v)} = 3$ and $l(v) \in P(v)$. Hence we get that
    \begin{align*}
        \sum_{u \in T'\cup T_i} \frac{H(w''(u))}{|T' \cup T_i|} &\leq \frac{\sum_{\ell \in T_i \setminus \{r_i\}} H(w''(\ell)) + 1/3  + \sum_{u \in T'} H(w'(u))}{|T' \cup T_i|}\\
        &< \frac{1.8917 (|T_i|-1)  + 1.8917|T'|}{|T' \cup T_i|}  = 1.8917.
    \end{align*}
\end{proof}


\subsection{The proof of Theorem~\ref{thm:witness-tree}}

We are now ready to put everything together and prove Theorem~\ref{thm:witness-tree}.

\begin{proof}[Proof of Theorem~\ref{thm:witness-tree}]
Let $G = (R \cup S, E)$ be a CA-Node-Steiner-Tree instance, where $R$ is the set of terminals. Let $T = (R \cup S^*, E^*)$ be an optimal solution in which the terminals are leaves. By the discussion of Section~\ref{sec:removing-terminals}, it suffices to construct a tree $W = (F, E_W)$ spanning the set of final Steiner nodes of $T$, whose imposed $w$-value on $v \in S^*$ is defined as
\begin{equation*}
    w(v) \coloneqq |\{(p,q) \in E_W: v \textrm{ is a node of the path between }p \textrm{ and }q \textrm{ in } T[S^*]\}| + \mathbbm{1}[v \in F].
\end{equation*}
Following the discussion of Section~\ref{sec:decomposition}, we root $T$ at an arbitrary leaf $r$, and decompose it into a set of final-components $T_1, \ldots, T_\tau$, with the property that $r \in T_1$ and for every $i \in [\tau]$, $\bigcup_{j \leq i} T_j$ is connected. We process the trees in increasing index order and do the following. We maintain a tree $W'$ spanning the final Steiner nodes of $T' = \bigcup_{j < i} T_j$, and during iteration $i$, we compute a tree $W_i$ spanning the final Steiner nodes of $T_i$ and then merge it with $W'$, as discussed in Section~\ref{sec:main-analysis}. The resulting tree indeed is a tree spanning the final Steiner nodes of $\bigcup_{j \leq i} T_j$, as demonstrated in Section~\ref{sec:main-analysis}. The analysis we did now shows that after each iteration $i$, the resulting tree spanning the final Steiner nodes of $\bigcup_{j \leq i} T_j$ satisfies the desired bound of $1.8917$. More precisely, by setting $\gamma' = 1.8917$ in Lemma~\ref{lem:h3} and by using Lemma~\ref{lem:finalbound}, we get that after each iteration, the resulting tree has an average $H(w)$-value that is strictly smaller than $1.8917$. Thus, a straightforward induction shows that the final tree $W$ will have an average $H(w)$-value that is strictly smaller than $1.8917$ and will span the final Steiner nodes of $T$. This means that $\gamma_G < 1.8917$, which implies that $\gamma \leq 1.8917$.
\end{proof}

%% file: fig7.tex
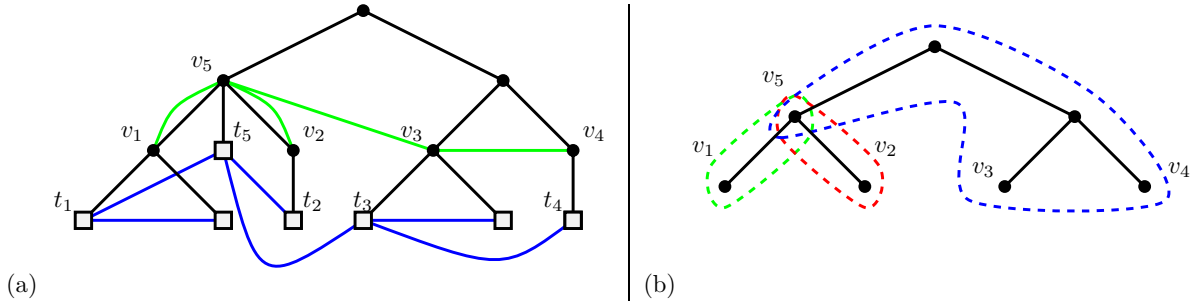
\begin{figure}[ht]
    \centering
    \scalebox{0.93}{
    
    \begin{tabular}{c|c}
        (a)
        \begin{tikzpicture}
            \draw[color=green,very thick]
                (1,1) node{} .. controls (1.25,1.65) .. (2,2) node{}
                
                (5,1) node{} -- (7,1) node{}
                (2,2) .. controls(2,2) .. (5,1)
                (2,2) .. controls (2.75,1.65) .. (3,1);
                
            \draw[color= blue,very thick] 
                (0,0) .. controls (0,0) .. (2,1)
                (2,1) .. controls (2.5,-1) .. (4,0)
                (2,1) .. controls (2,1) .. (3,0)
                (4,0) .. controls (6,-3/4) .. (7,0)
                (0,0) .. controls (0,0) .. (2,0)
                (4,0) .. controls (6,0) .. (6,0);
  
              \filldraw[very thick,color=black]
                (0,0)
            --  (1,1)
            
                (1,1)
            --  (2,0)
            
                (3,1)
            --  (3,0)
            
                (5,1)
            --  (4,0)
                (5,1)
            --  (6,0)
            
                (7,1)
            --  (7,0) node[draw=black, fill=black!5]{}
            
                (2,2)
            --  (2,1)
            
                (1,1) circle(2pt) node[align=left,   anchor = south east] {$v_1$}
            --  (2,2) circle(2pt) node[align=left,   anchor =  south east] {$v_5$} 
            --  (4,3) circle(2pt) 
            --  (6,2) circle(2pt) 
            --  (7,1) circle(2pt) node[align=left,   anchor = south west] {$v_4$} 
                (2,2)
            --  (3,1) circle(2pt) node[align=left,   anchor = south west] {$v_2$}
                (6,2)
            --  (5,1) circle(2pt) node[align=right,   anchor = south east]{$v_3$}
            
                (0,0) node[draw=black, fill=black!5]{}
                (0,0) node[align=left,   anchor = south east ] {$t_{1}$}
                
                (2,0) node[draw=black, fill=black!5]{}
                
                (3,0) node[draw=black, fill=black!5]{}
                (3,0) node[align=left,   anchor = south west] {$t_{2}$}
                
                (4,0) node[draw=black, fill=black!5]{}
                (4,0) node[align=left,   anchor = south] {$t_{3}$}
                
                (6,0) node[draw=black, fill=black!5]{}
                
                (7,0) node[draw=black, fill=black!5]{}
                (7,0) node[align=left,   anchor = south east] {$t_{4}$}
                
                (2,1) circle(2pt) node[draw=black, fill=black!5]{}
            
                (2,1) node[align=right,   anchor = south west]{$t_5$};

        \end{tikzpicture}
         &
         (b)
         \begin{tikzpicture}
            \draw[very thick,color=green,dashed]
                plot [smooth cycle] coordinates {(1,0.7 ) (0.8,1.3 ) (2,2.3 ) (2.2,1.7 )};
            
            \draw[very thick,color=red,dashed]
                plot [smooth cycle] coordinates {(3,0.7) (3.2,1.3) (2,2.3) (1.8,1.7)};
            
            \draw[very thick,color=blue,dashed]
                plot [smooth cycle] coordinates {(2,2.3) (1.8,1.7) (4.3,2.2) (4.5,0.8) (7.3,0.8) (6.3,2.3) (4,3.3)};

            \filldraw[very thick]
                (1,1) circle(2pt) 
            --  (2,2) circle(2pt) 
            --  (4,3) circle(2pt)
            --  (6,2) circle(2pt)
            --  (7,1) circle(2pt) 
                (1,1.3) node[align=left,   anchor = south east] {$v_1$}
                (2,2.3) node[align=left,   anchor =  south east] {$v_5$} 
                (7.2,1) node[align=left,   anchor = south west] {$v_4$}

                (2,2) 
            --  (3,1) circle(2pt) 
                (3,1.3) node[align=left,   anchor = south west] {$v_2$}
            
                (6,2) 
            --  (5,1) circle(2pt) 
                (5,1) node[align=right,   anchor = south east]{$v_3$};
                
        \draw[color=white]
                (7,-1/2);
            
        \end{tikzpicture}
    \end{tabular}
    
    }
    
\caption{(a) Given is a Steiner tree, marked with black, where the terminals are marked as square nodes. The set of final nodes is $\{v_1,v_2,v_3,v_4,v_5\}$, and a tree spanning them is denoted with the green edges. In order to obtain a terminal spanning tree, each green edge $(v_i,v_j)$ is replaced with the edge $(t_i,t_j)$,  marked with blue, where $t_i= \mathtt{rep}(v_i)$ for every $i\in[5]$. Finally, every terminal with a ``sibling" terminal has a corresponding edge added, marked again with blue. (b) An example of the the Steiner tree from (a) decomposed into its final-components as defined in Section~\ref{sec:decomposition}.}
\label{fig:Removing-terminals}
\end{figure}

%% file: fig2.tex
\begin{figure}[ht!]
\centering
    \begin{tabular}{c|c|c}
    (a)
        \begin{tikzpicture}[scale=0.75]
            \filldraw[very thick, color=red]
                (0,0) node{} -- (1,1) node{}
                (1,1) node{} -- (2.5,2) node{}
                (4,1) node{} -- (3.5,0) node{}
                (2,0) node{} -- (1,-1) node{};
                
            \draw[very thick,color=black]
                
                (2.5,3) circle circle (2pt)node[align=left,   anchor= south west]{$1$} --
                (2.5,2) circle (2pt){} -- (4,1) circle (2pt){}
                                     -- (4.5,0) circle (2pt)node[align=left,   anchor= south west]{$5$}
                                        (4,1) circle (2pt){}
                                        (3.5,0) circle (2pt)node[align=left,   anchor= south east]{$4$}
                                        
                (0,0) circle (2pt)node[align=left,   anchor= south east]{$2$}  
                
                (1,0) circle (2pt)node[align=left,   anchor= south east]{$3$}
                    -- (1,1) circle (2pt){}
                    -- (2,0) circle (2pt){}
                    -- (2,-1) circle (2pt)node[align=left,   anchor= south east]{$7$}
                    
                (2,0) node{} -- (3,-1) circle (2pt)node[align=left,   anchor= south west]{$8$}
                (1,-1) circle (2pt)node[align=left,   anchor= south east]{$6$};

            \filldraw[ultra thick,color=red]
                (0,0) circle (2pt) node[rectangle, draw=red!60, fill=red!5, very thick, minimum size=1.1mm]{}
                (1,0) circle (2pt) node[rectangle, draw=red!60, fill=red!5, very thick, minimum size=1.1mm]{}
                (1,-1) circle (2pt) node[rectangle, draw=red!60, fill=red!5, very thick, minimum size=1.1mm]{}
                (2,-1) circle (2pt) node[rectangle, draw=red!60, fill=red!5, very thick, minimum size=1.1mm]{}
                (3,-1) circle (2pt) node[rectangle, draw=red!60, fill=red!5, very thick, minimum size=1.1mm]{}
                (3.5,0) circle (2pt) node[rectangle, draw=red!60, fill=red!5, very thick, minimum size=1.1mm]{} 
                (4.5,0) circle (2pt) node[rectangle, draw=red!60, fill=red!5, very thick, minimum size=1.1mm]{}
                (2.5,3) circle (2pt) node[rectangle, draw=red!60, fill=red!5, very thick, minimum size=1.1mm]{};

        \end{tikzpicture}
        & 
        
        (b)
        \begin{tikzpicture}[scale=0.75]
            \draw[ultra thick,color=black]
                
                (6,4) node[align=left,  anchor= south west]{$1$} 
            --  (6,3) circle (2pt)node[align=left,   anchor= south west]{$2$} 
            --  (7.75,1) circle (2pt)node[align=left,   anchor= south west]{$4$}
            --  (8,0) circle (2pt)node[align=left,   anchor= south west]{$5$}
                                     
                (4,1) circle (2pt)node[align=left,   anchor= south east]{$3$}
            --  (6,3) node{} 
            -- (6,1) circle (2pt)node[align=left,   anchor= south east]{$6$}
            -- (5,0) circle (2pt)node[align=left,   anchor= south east]{$7$}
                
                (6,1) circle (2pt) -- (7, 0) circle (2pt)node[align=left,   anchor= south west]{$8$};
                
            \filldraw[ultra thick,color=red]
                (6,4) circle (2pt) node[rectangle, draw=red!60, fill=red!5, very thick, minimum size=1.1mm]{}
                (6,3) circle (2pt) node[rectangle, draw=red!60, fill=red!5, very thick, minimum size=1.1mm]{}
                (7.75,1) circle (2pt) node[rectangle, draw=red!60, fill=red!5, very thick, minimum size=1.1mm]{}
                (8,0) circle (2pt) node[rectangle, draw=red!60, fill=red!5, very thick, minimum size=1.1mm]{}
                (6,1) circle (2pt) node[rectangle, draw=red!60, fill=red!5, very thick, minimum size=1.1mm]{}
                (5,0) circle (2pt) node[rectangle, draw=red!60, fill=red!5, very thick, minimum size=1.1mm]{}
                (7,0) circle (2pt) node[rectangle, draw=red!60, fill=red!5, very thick, minimum size=1.1mm]{}
                (4,1) circle (2pt) node[rectangle, draw=red!60, fill=red!5, very thick, minimum size=1.1mm]{};
                
        \end{tikzpicture}
        & 
        (c)
        \begin{tikzpicture}[scale=0.75]
            \draw[very thick]
                (0,0) node{} -- (1,1) node{}
                             -- (2.5,2) node{}
                (4,1) node{} -- (3.5,0) node{}
                (2,0) node{} -- (1,-1) node{};
                
            \draw[very thick,color=black]
                
                (2.5,3) circle circle (2pt)node[align=left,   anchor= south west]{$1$} --
                (2.5,2) circle (2pt){} -- (4,1) circle (2pt){}
                                     -- (4.5,0) circle (2pt)node[align=left,   anchor= south west]{$5$}
                                        (4,1) circle (2pt){}
                                        (3.5,0) circle (2pt)node[align=left,   anchor= south east]{$4$}
                                        
                (0,0) circle (2pt)node[align=left,   anchor= south east]{$2$}  
                
                (1,0) circle (2pt)node[align=left,   anchor= south east]{$3$}
                    -- (1,1) circle (2pt){}
                    -- (2,0) circle (2pt){}
                    -- (2,-1) circle (2pt)node[align=left,   anchor= south east]{$7$}
                    
                (2,0) node{} -- (3,-1) circle (2pt)node[align=left,   anchor= south west]{$8$}
                (1,-1) circle (2pt)node[align=left,   anchor= south east]{$6$};
                
            \draw[color=blue,very thick]
                (0,0) node{} -- (1,0) node{}
                (0,0) node{} .. controls (1,2) .. (2.5,3) node{}
                (0,0) node{} .. controls (1/2,-1) .. (1,-1) node{}
                (0,0) node{} .. controls (2/3,1) .. (3.5,0) node{}
                (3.5,0) node{} .. controls (4,0).. (4.5,0) node{}
                (1,-1) node{} .. controls (1,-1) .. (2,-1) node{}
                (1,-1) node{} .. controls (2,-3/2) .. (3,-1) node{};                
                
            \draw[ultra thick] 
                (0,0) circle (2pt) node[rectangle, draw=black!60, fill=black!5, very thick, minimum size=1mm]{}
                (1,0) circle (2pt) node[rectangle, draw=black!60, fill=black!5, very thick, minimum size=1.1mm]{}
                (1,-1) circle (2pt) node[rectangle, draw=black!60, fill=black!5, very thick, minimum size=1.1mm]{}
                (2,-1) circle (2pt) node[rectangle, draw=black!60, fill=black!5, very thick, minimum size=1.1mm]{}
                (3,-1) circle (2pt) node[rectangle, draw=black!60, fill=black!5, very thick, minimum size=1.1mm]{}
                (3.5,0) circle (2pt) node[rectangle, draw=black!60, fill=black!5, very thick, minimum size=1.1mm]{} 
                (4.5,0) circle (2pt) node[rectangle, draw=black!60, fill=black!5, very thick, minimum size=1.1mm]{}
                (2.5,3) circle (2pt) node[rectangle, draw=black!60, fill=black!5, very thick, minimum size=1.1mm]{};

        \end{tikzpicture}
    \end{tabular}
     \caption{(a) Given is a graph $G=(V,E)$, with final Steiner nodes depicted as red squares. At each interior node $u$ we mark in red the $u$-to-leaf-path $P(u)$ with the least weight (sum of the inverse of the node degrees on the path); this corresponds to step (3) of Algorithm~\ref{alg:computing_w}. (b) At step (4) of Algorithm~\ref{alg:computing_w}, we contract the red edges and label the new contracted nodes with the name of the original (marked) leaf node. (c) A pictorial description of the witness tree obtained for the original graph.}
     \label{fig:deterministic-witness-trees}
\end{figure}
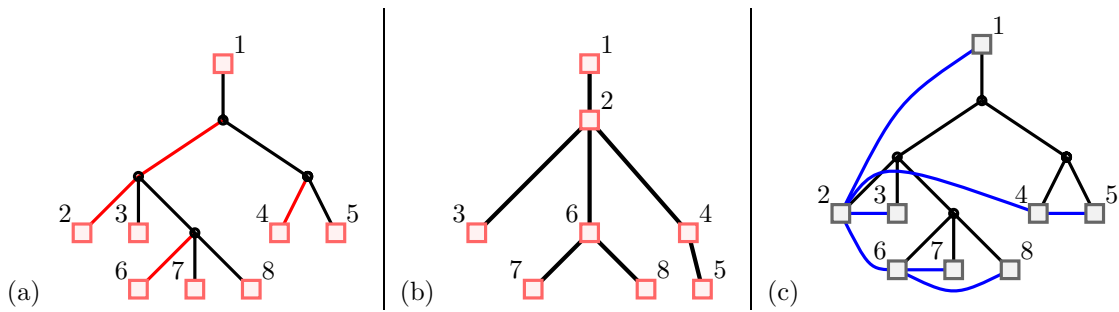

%% file: fig1.tex
\begin{figure}[ht]
\centering

    \begin{tikzpicture}
        \draw [color=red, very thick]
            (0,0) .. controls (5.5,-0.7) .. (5,3/2);
          \filldraw[very thick]
            (0, 0)  node[draw=black, fill=black!5]{} node[align=left, anchor = south east] {$\ell(u)$}
        --  (1,1/2) circle(2pt) node[align=left, anchor = south east] {$u_1$}
        --  (2,2/2) circle(2pt) node[align=left, anchor = south east] {$u$}
        --  (3,3/2) circle(2pt) node[align=left, anchor = south east] {$\mathtt{parent}(u)$}
        --  (4,4/2) circle(2pt) node[align=left, anchor = south east] {$a(u)$}
        --  (5,5/2) circle(2pt) node[align=left, anchor = west] {$r_i$}node[draw=black, fill=black!5]{}
        
            (4,2) 
        --  (5,3/2) node[draw=black, fill=black!5]{}
            (5.2,1.5)  node[align=left, anchor = south ] {$\ell(a(u))$} 
            
            (1,1/2)
        --  (2,0/2) node[draw=black, fill=black!5]{}
        
            (2,2/2) 
        --  (3,1/2) circle(2pt) 
        --  (4,0/2) node[draw=black, fill=black!5]{};
    \end{tikzpicture}

    \endpgfgraphicnamed
    
    \caption{Given is an example of a node $u$; the minimum-weight path from $u$ to the leaf $\ell(u)$,  $P(u) = \{u,u_1,\ell(u)\}$; the parent of $u$, $\mathtt{parent}(u)$; the first node on the $u$ to $r_i$ path, $a(u)$, such that $\ell(u) \neq \ell(a(u))$; and finally, the edge $e_u = \{\ell(u), \ell(a(u))\}.$}
        
\label{fig:definitions}
\end{figure}
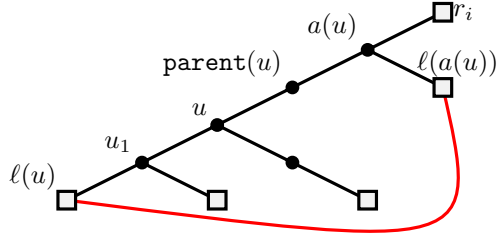

%% file: fig5.tex
\begin{figure}[ht]
\centering
      \begin{tikzpicture}[scale=0.8]
                
                \draw[color=red,very thick]
                    (-5,-4) .. controls(-4,-4.75) ..  (-2,-4)
                    (-5,-4) .. controls(-4,-4.85) ..  (1,-4)
                    (-5,-4) .. controls(-4,-4.95) ..  (4,-4)
                    
                    (-5,-4) .. controls(-6.5,-3.6) .. (-3,-.65)
                    
                    (-4,-4) -- (-5,-4)
                    (-1,-4) -- (-2,-4)
                    (1,-4) -- (2,-4)
                    (4,-4) -- (5,-4);
                    
                \filldraw[very thick]
                    (0,-1.35) circle(2pt) -- (4.5,-3) circle(2pt) 
                    (0,-1.35) -- (1.5,-3) circle(2pt)
                    (0,-1.35) -- (-1.5,-3) circle(2pt)
                    (0,-1.35) -- (-4.5,-3) circle(2pt)
                    
                    (-4.5,-3) -- (-4,-4) node[draw=black, fill=black!5]{}
                    (-4.5,-3) -- (-5,-4) node[draw=black, fill=black!5]{}
                    
                    (-1.5,-3) -- (-1,-4) node[draw=black, fill=black!5]{}
                    (-1.5,-3) -- (-2,-4) node[draw=black, fill=black!5]{}
                    
                    (1.5, -3) -- (1,-4) node[draw=black, fill=black!5]{}
                    (1.5, -3) -- (2,-4) node[draw=black, fill=black!5]{}
                    
                    (4.5, -3) -- (4,-4) node[draw=black, fill=black!5]{}
                    (4.5, -3) -- (5,-4) node[draw=black, fill=black!5]{};

                \draw[very thick, dashed]
                    (0,-1.35) -- (-3,-.65);

                \draw
                    (0.2,-1) node{$u$}
                    (-4.8,-2.8) node{$u_1$}
                    (-1.8,-2.8) node{$u_2$}
                    (1.7,-2.8) node{$u_3$}
                    (4.65,-2.8) node{$u_4$}
                    
                    (-4.5,-1.5) node{$e_u$}
                    (-2.8,-3.9) node{$e_{u_2}$}
                    (0.2,-3.9) node{$e_{u_3}$}
                    (3.2,-3.9) node{$e_{u_4}$};

      \end{tikzpicture}
        \endpgfgraphicnamed
        
        \caption{Given is the tree $Q_u$ of $T_i$. $W^u$ contains all of the red edges. Note that each $W^{u_i}$ is equal to the red edge with both endpoints below $u_i$ plus $e_{u_i}$, for $i = 1,2,3,4$. We can also see an example of the first three parts of Lemma~\ref{lem:increase}: (a) $w^u(u) = 4$; (b) for $j = 2,3,4$, and every $\ell\in Q_{u_j}$, $w^u(\ell) = w^{u_j}(\ell)$, and; (c) for each $\ell\in Q_{u_1}\setminus P(u_1), w^u(\ell) = w^{u_1}(\ell)$.}   
\label{fig:wu-description}
\end{figure}
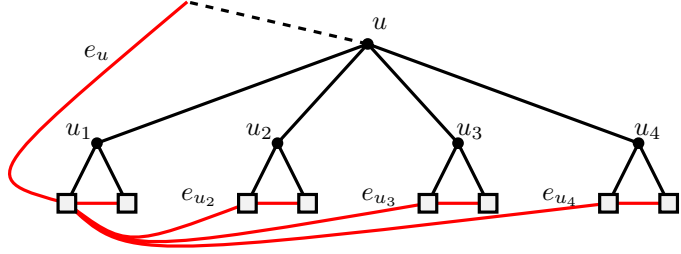


%% file: leaf_adjacent_block_tap.tex
\section{Improved approximation for leaf-adjacent Block-TAP}
\label{sec:leaf-adjacent-upper-bound}

In this section, we consider Block-TAP instances where at least one endpoint of every link is a leaf; we call such instances \emph{leaf-adjacent} Block-TAP instances. We note that leaf-adjacent Block-TAP is a generalization of the case where both endpoints are leaves, often called \emph{leaf-to-leaf} Block-TAP. For leaf-to-leaf Block-TAP, Nutov gave a $1.6\bar{6}$-approximation~\cite{DBLP:journals/corr/abs-2009-13257}. Here, we give a $(1.8\bar{3} + \varepsilon)$-approximation for the more general setting of leaf-adjacent Block-TAP.

Throughout this section, we assume that we are given a Block-TAP instance, where $T=(V_T, E_T)$ is the input tree, $R_T \subseteq V_T$ is its set of leaves and $L \subseteq \binom{V_T}{2}$ is the set of links, such that $\ell \cap R_T \neq \emptyset$ for every $\ell \in L$. Using the reduction to CA-Node-Steiner-Tree (see Theorem~\ref{cor:approx-preserve-block}), we get a CA-Node-Steiner-Tree instance $G= (R\cup S,E)$ that satisfies the following property: every Steiner node in $G$ is adjacent to at least one terminal. We call such instances \emph{leaf-adjacent} CA-Node-Steiner-Tree instances.

We start by observing that any leaf-adjacent CA-Node-Steiner-Tree instance $G = (R \cup S, E)$ has an optimal solution $T^* = (R \cup S^*, E^*)$ such that every Steiner node $s \in S^*$ is adjacent to at least one terminal in $T^*$. To see this, suppose that there is a Steiner node $s \in S^*$ that is not adjacent to any terminal in $T^*$. Since $s$ is adjacent to some terminal $r \in R$ in $G$, we add the edge $(s,r)$ to $T^*$. This creates a cycle that goes through $s$. By removing the edge of the cycle that is adjacent to $s$ and is not equal to $(s,r)$, we end up with a different optimal solution where $s$ is adjacent to a terminal. This shows that we can transform any optimal solution to a solution that satisfies the desired property, namely, that every Steiner node is adjacent to a terminal. Thus, from now on we assume that $T^*$ satisfies this property.

Our analysis consists of demonstrating that the procedure laid out in Section~\ref{sec:main-analysis} finds a witness tree $W$ for $T^*$ such that if $w: S^* \to \mathbb{N}$ is the vector imposed on $S^*$ by $W$, then we have $\E[H(w)] \leq H(3)$. To see this, we first transform $T^*$ to a forest as follows. In case there is a terminal $r \in R$ of $T^*$ that is not a leaf, we split the tree $T^*$ at $r$ by first removing $r$, and then adding back a copy of $r$ as a leaf to the $d(r)$ trees of the forest $T^* \setminus \{r\}$, where $d(r)$ is the degree of $r$ in $T^*$. By performing this operation for all terminals whose degree is larger than $1$, we end up with a forest, where each terminal $r \in R$ appears in $d(r)$ trees, in total, such that for each tree of the forest, every final Steiner node is adjacent to a terminal and moreover, all terminals are leaves. Since each Steiner node belongs to exactly one tree of this forest, it is easy to see that if we compute a terminal spanning tree for each tree of the resulting forest, and take the union of these trees, then the $w$-value of each Steiner node imposed by the final terminal spanning tree is the same as the $w$-value imposed by the terminal spanning tree computed for the particular tree of the forest which the Steiner node belongs to. Thus, without loss of generality, we assume from now on that $T^*$ satisfies both desired properties, namely that every Steiner node is adjacent to a terminal and every terminal is a leaf.

We now remove the terminals from $T^*$ as in Section~\ref{sec:removing-terminals} and decompose the remaining tree $T^*[S^*]$, into final-components as in Section~\ref{sec:decomposition}, which we denote $T^*_1,\dots, T^*_\tau$. It then remains to see that when we construct the witness trees $\{W_i\}_{i=1}^\tau$ and merge them together to create the witness tree $W$, every tree in $\{T^*_i\}_{i=1}^\tau$ falls under Case 1 of the analysis. In particular, this implies that we only apply Lemma~\ref{lem:h3} when merging witness trees. Since the guarantee of Lemma~\ref{lem:h3} holds for any $\gamma' \geq H(3)$, we can use it with $\gamma' = H(3)$, thus giving the desired bound. 

More precisely, observe that each node of $T^*[S^*]$ is a final node by the construction of $T^*$, so the final-components $T^*_1,\dots, T^*_\tau$ are in fact the edges of $T^*[S^*]$. When the analysis of Section~\ref{sec:witness-tree} considers a fixed final-component $T^*_i$ rooted at $r_i$ with $(r_i,v)$ being the unique edge incident to $r_i$ inside $T^*_i$, we know that $v$ is a final node by construction, so we only consider Case 1 when analyzing the change of the average $H(w)$-value. Thus, we obtain the following theorem.

\begin{theorem}\label{thm:leaf-adjacent}
Given is a leaf-adjacent CA-Node-Steiner-Tree instance $G = (R \cup S, E)$, and let $T^* = (R \cup S^*, E^*)$ be an optimal Steiner tree such that each Steiner node $s \in S^*$ is adjacent to at least one terminal in $T^*$. Then, there exists a witness tree $W$ such that $\E[H(w)]\leq H(3)$, where $w$ is the vector imposed on $S^*$ by $W$.
\end{theorem}

The witness tree that achieves the stated value of $\E[H(w)]$ is in fact a very natural one. We perform the operation above to transform $T^*$ into a forest where the terminals are exactly the leaves of $T^*$. For every component, and for every edge $(u,v)$ of that component, there is a corresponding edge in the witness tree between one of the terminals adjacent to $u$ and one adjacent to $v$. Furthermore, if two terminals have the same adjacent Steiner node, then the witness tree has an edge between those two terminals as well.

The above theorem, along with Theorems~\ref{cor:approx-preserve-block} and~\ref{thm:main-approx}, immediately imply Theorem~\ref{thm:approx-leaf-adjacent-block-tap}.

%% file: lower_bound.tex
\section{Limitations of the witness tree analysis}
\label{sec:lower-bound}

In this section, we show that we cannot get an approximation factor better than $1.8\bar{3}$ by using (deterministic/randomized) witness trees in the analysis of Algorithm~\ref{alg:rounding-scheme}. Moreover, we show that Algorithm~\ref{alg:computing_w} sometimes finds a witness tree that gives an average $H(w)$-value greater than $1.8504$. This shows that our analysis of the approximation factor of Algorithm~\ref{alg:rounding-scheme} with the witness tree generated by Algorithm~\ref{alg:computing_w} is off by at most $0.0416$.

\subsection{Lower bound for $\gamma$}

In this section, we show that $\gamma \geq H(3)$. To prove this, we construct a family of leaf-adjacent CA-Node-Steiner-Tree instances that are trees, i.e., the optimal Steiner tree is the input graph itself, along with a specific witness tree that we prove to be optimal, such that the average $H(w)$-value is at least $H(3) - \varepsilon$, for any fixed $\varepsilon > 0$. More formally, we prove the following theorem.


\begin{theorem}\label{thm:bad-example-H(3)}
For any $\varepsilon > 0$, there exists a leaf-adjacent CA-Node-Steiner-Tree instance $G_\varepsilon$ such that $\gamma_{G_\varepsilon} > H(3) - \varepsilon$.
\end{theorem}

\begin{proof}
We give an explicit construction that corresponds to a leaf-to-leaf Block-TAP instance, i.e., every Steiner node is adjacent to exactly two terminals. Recall that leaf-to-leaf instances are a specific case of leaf-adjacent instances. Consider the following graph.  Let $t = \lceil \frac{2}{3\varepsilon} \rceil + 1$ and let $S = \{s_1, \ldots, s_t\}$. Let $R = \bigcup_{i = 1}^t\{r_i^{(1)}, r_i^{(2)}\}$. For every $i \in [t-1]$, the node $s_i$ is adjacent to $s_{i+1}$. In other words, the nodes $s_1, \ldots, s_t$ form a path whose endpoints are $s_1$ and $s_t$. Finally, every Steiner node $s_i$ is adjacent to terminals $r_i^{(1)}$ and $r_i^{(2)}$. Let $G_\varepsilon = (R \cup S, E)$ be the resulting CA-Node-Steiner-Tree instance. Observe that there is a unique optimal Steiner tree, which is the graph $G_\varepsilon$ itself. This implies that
\begin{equation*}
    \gamma_{G_\varepsilon} = \frac{1}{t} \min_{\substack{W:\; W\textrm{ is }\\\textrm{witness tree}}} \;\left(\sum_{i = 1}^t H(w(s_i)) \right).
\end{equation*}

We now consider the following witness tree $W$, that simply ``follows" the path $s_1, \ldots, s_t$. More precisely, we define $W = (R, E_W)$ as follows:
\begin{itemize}
    \item There is an edge between $r_i^{(1)}$ and $r_i^{(2)}$ for every $i \in [t]$.
    \item There is an edge between $r_i^{(1)}$ and $r_{i+1}^{(1)}$ for every $i \in [t-1]$.
\end{itemize}
It is easy to see that $w(s_1) = w(s_t) = 2$, and  $w(s_i) = 3$ for every $i \in \{2, \ldots, t-1\}$, and so we get that
\begin{equation*}
    \frac{1}{t} \sum_{i = 1}^t H(w(s_i)) =  \frac{2H(2) + (t-2) H(3)}{t} = H(3) - \frac{2}{3t} > H(3) - \varepsilon.
\end{equation*}

Thus, the only thing remaining to show is that $W$ is an optimal witness tree with respect to minimizing $\sum_{i = 1}^t H(w(s_i))$. For that, let $W^*$ be an optimal witness tree with respect to minimizing $\sum_{i = 1}^t H(w(s_i))$, and let $w^*$ be the vector imposed on $S$ by $W^*$. We first check if $(r_i^{(1)}, r_i^{(2)}) \in W^*$ for every $i \in [t]$. Suppose that there is an $i \in [t]$ such that $(r_i^{(1)}, r_i^{(2)}) \notin W^*$. We now add the edge $(r_i^{(1)}, r_i^{(2)})$ to $W^*$ and a cycle is created. Clearly, there exists an edge $e' \in W^*$ adjacent to $r_i^{(2)}$ such that $W' = (W^* \setminus \{e'\}) \cup \{(r_i^{(1)}, r_i^{(2)})\}$ is a terminal spanning tree. It is also easy to see that $w'(s_j) \leq w^*(s_j)$ for every $j \in [t]$, where $w'$ is the vector imposed on $S$ by $W'$. We conclude that $\sum_{i = 1}^t H(w'(s_i)) \leq \sum_{i = 1}^t H(w^*(s_i))$, and so from now on we assume without loss of generality that $(r_i^{(1)}, r_i^{(2)}) \in W^*$ for every $i \in [t]$.

We now impose some more structure on $W^*$. In particular, we process the terminals in increasing order of index, and for each $i \in [t - 1]$, we replace any edge of the form $(r_i^{(2)}, r_j^{(x)})$, $j > i$ and $x \in \{1,2\}$, with the edge $(r_i^{(1)}, r_j^{(1)})$. It is easy to see that no $w$-value changes, and so from now on we assume without loss of generality that for every $i \in [t]$, the terminal $r_i^{(2)}$ is a leaf of $W^*$ that is connected with $r_i^{(1)}$.

Finally, we turn to the edges $(r_i^{(1)}, r_{i+1}^{(1)})$ that are existent in $W$ for every $i \in [t-1]$. If $W^* \neq W$, then there exist $1 \leq i < j \leq t$, with $j - i > 1$ such that $e = (r_i^{(1)}, r_j^{(1)}) \in W^*$. We remove $e$ from $W^*$, and get two subtrees $W_1^*$ and $W_2^*$, where $r_i^{(1)} \in W_1^*$ and $r_j^{(1)} \in W_2^*$. Without loss of generality, we assume that $r_{i + 1}^{(1)}\in W_1^*$; the case of $r_{i + 1}^{(1)}\in W_2^*$ is handled in a similar way. We now add the edge $e' = (r_{i+1}^{(1)}, r_j^{(1)})$ and obtain a new witness tree $W' = W_1^*\cup W_2^*\cup \{e'\}$. Note that the removal of the edge $e$ causes $w(s_k)$ to decrease by $1$ for every Steiner node $s_k$ with $i \leq k \leq j$, while the addition of the edge $e'$ increases $w(s_k)$ by $ 1$ for every Steiner node $s_k$ with $i+1 \leq k \leq j$. We conclude that $\sum_{i = 1}^t w'(s_i) \leq \sum_{i = 1}^t w^*(s_i)$, where $w'$ is the vector imposed on $S$ by $W'$. Putting everything together, we conclude that $W$ is an optimal witness tree with respect to minimizing $\sum_{i = 1}^t H(w(s_i))$, and so we get that $\gamma_{G_\varepsilon} > H(3) - \varepsilon$.
\end{proof}

An immediate corollary of the above theorem is the following.
\begin{corollary}
$\gamma \geq H(3) = 1.8\bar{3}$.
\end{corollary}

The above lower bound shows a limitation of our techniques. More precisely, it demonstrates that one cannot hope to obtain a better than $1.8\bar{3}$-approximation by selecting a different (deterministic/randomized) witness tree.

\subsection{Lower bound for the witness tree generated by Algorithm~\ref{alg:computing_w}}

Finally, besides the $1.8\bar{3}$ lower bound, we also give an explicit example of a CA-Node-Steiner-Tree instance for which Algorithm~\ref{alg:computing_w} finds a witness tree that gives an average $H(w)$-value greater than $1.8504$. This shows that our analysis of the approximation factor of Algorithm~\ref{alg:rounding-scheme} with the witness tree generated by Algorithm~\ref{alg:computing_w} is off by at most $0.0416$.

We now describe the CA-Node-Steiner-Tree instance $T = (R\cup S^*, E^*)$, which is a tree; thus the optimal solution is $T$ itself. We will show that Algorithm~\ref{alg:computing_w} finds a witness tree spanning the terminals of $T$ that gives an average $H(w)$-value greater than $1.8504$. We now describe the tree $T$, which consists of five layers:
\begin{itemize}
    \item the $1^{\mathrm{st}}$ layer contains a single node $r$.
    \item the $2^{\mathrm{nd}}$ layer consists of $9$ nodes $\{x_1, \ldots, x_9\}$. For each $i \in [9]$, there is an edge $(r, x_i)$.
    \item the $3^{\mathrm{rd}}$ layer consists of $9$ groups of $5$ nodes, denoted as $\bigcup_{i = 1}^9 \{y_{i1}, \ldots, y_{i5}\}$. For each $i \in [9]$ and $j \in [5]$, there is an edge $(x_i, y_{ij})$.
    \item the $4^{\mathrm{th}}$ layer consists of $9 \times 5$ groups of $4$ nodes, denoted as $\bigcup_{i = 1}^9 \bigcup_{j = 1}^5 \{z_{ij1}, \ldots, z_{ij4}\}$. For each $i \in [9]$, $j \in [5]$ and $k \in [4]$, there is an edge $(y_{ij}, z_{ijk})$.
    \item the $5^{\mathrm{th}}$ layer consists of $9 \times 5 \times  4 \times 2$ terminals, denoted as $R = \bigcup_{i = 1}^9 \bigcup_{j = 1}^5 \bigcup_{k = 1}^4 \{q_{ijk}^{(1)},q_{ijk}^{(2)}\}$. For each $i \in [9]$, $j \in [5]$ and $k \in [4]$, there is an edge $(z_{ijk}, q_{ijk}^{(1)})$ and an edge $(z_{ijk}, q_{ijk}^{(2)})$.
\end{itemize}
There are no other edges contained in $T$. As in Section~\ref{sec:removing-terminals}, the set of final Steiner nodes is the set of vertices in the $4^{\mathrm{th}}$ layer, and since there are exactly two distinct terminals adjacent to each final Steiner node, we remove all the terminals and simply increase the $w$-value of the each final Steiner node by $1$. A pictorial representation of $T \setminus R$ is provided in Figure~\ref{fig:tree-define}.
\input{./fig4.tex}

From now on, we denote the set of vertices of the subtree rooted at $x_i$ as $Q_{x_i}$. To use Algorithm~\ref{alg:computing_w}, we root $T$ at a final Steiner node, and without loss of generality, by the symmetry of $T$, we can root $T$ at $z_{222}$. Similarly, without loss of generality we can let the ordering of the leaves of $T$ selected at step (1) of Algorithm~\ref{alg:computing_w}  coincide with the lexicographic ordering of the leaf indices. With this ordering, for each non-final Steiner node $u\in T$, we find the minimum-weight path $P(u)$ described at step (3) of Algorithm~\ref{alg:computing_w}, which we highlight in Figure~\ref{fig:bad-example}, along with the corresponding $w$-values. 

\input{./fig3.tex}

For each $Q_{x_i}$, $i\in \{3,\dots,9\}$, we have
\begin{equation*}
    \sum_{v\in Q_{x_i}} H(w(v)) = 15H(2)+4H(4)+5H(5)+H(8)+H(9),
\end{equation*}
and for $Q_{x_1}$ we have 
\begin{equation*}
    \sum_{v\in Q_{x_1}} H(w(v)) = 15H(2)+4H(4)+4H(5)+H(12)+H(15)+H(16).
\end{equation*}
We now turn to the nodes in $T\setminus \bigcup_{i \in [9] \setminus \{2\}}Q_{x_i}$. It is not hard to see that $w(r) = 8$, and
\begin{equation*}
    \sum_{v\in T\setminus \cup_{i \in [9] \setminus \{2\}} Q_{x_i}} H(w(v)) = 15H(2)+4H(4)+5H(5)+2H(8)+H(9).
\end{equation*}
Putting everything together, we get that the average $H(w)$-value of the Steiner nodes of $T$ is 
\begin{align*}
    \frac{\sum_{v\in S^*} H(w(v))}{|S^*|} &= \frac{\sum_{v\in T\setminus \cup_{i \in [9] \setminus \{2\}} Q_{x_i}} H(w(v)) +  \sum_{i=3}^9 \sum_{v\in Q_{x_i}} H(w(v)) + \sum_{v\in Q_{x_1}} H(w(v)) }{235}\\
    &= \frac{135H(2) + 36H(4) + 44H(5) + 9H(8) + 8H(9) + H(12) + H(15) + H(16)}{235}\\
    &> 1.8504.
\end{align*}
We conclude that for this particular instance, the analysis of of Algorithm~\ref{alg:rounding-scheme} with the witness tree generated by Algorithm~\ref{alg:computing_w} will necessarily give an approximation factor strictly larger than $1.8504$.

%% file: fig4.tex
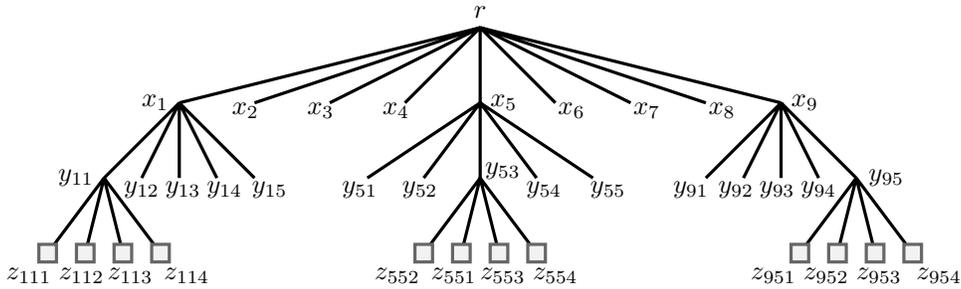
\begin{figure}[ht]
\centering
        \begin{tikzpicture}[]
            \draw[very thick]
                (0,0)  -- (-4,-1) node[anchor=east]{$x_1$}
                (0,0)  -- (-3,-1) 
                          (-2.8,-1.1) node[anchor=east]{$x_2$}
                (0,0)  -- (-2,-1)  
                          (-1.8,-1.1) node[anchor=east]{$x_3$}
                (0,0)  -- (-1,-1) 
                          (-0.8,-1.1) node[anchor=east]{$x_4$}
                (0,0) node[anchor=south]{$r$} -- (0,-1)node[anchor=west]{$x_5$}
                (0,0)  -- (1,-1) 
                          (0.9,-1.1)node[anchor=west]{$x_6$}
                (0,0)  -- (2,-1) 
                          (1.9,-1.1)node[anchor=west]{$x_7$}
                (0,0)  -- (3,-1) 
                          (2.9,-1.1)node[anchor=west]{$x_8$}
                (0,0)  -- (4,-1) node[anchor=west]{$x_9$}

                (-4,-1)  -- (-5,-2) node[anchor=east]{$y_{11}$}
                (-4,-1)  -- (-4.5,-2) 
                            (-4.5,-2.15) node{$y_{12}$}
                (-4,-1)  -- (-4,-2)
                            (-3.95,-2.15) node{$y_{13}$}
                (-4,-1)  -- (-3.5,-2) 
                            (-3.4,-2.15) node{$y_{14}$}
                (-4,-1)  -- (-3,-2) 
                            (-2.8,-2.15) node{$y_{15}$}
                
                (4,-1)  -- (5,-2) 
                           (5.4,-2) node{$y_{95}$}
                (4,-1)  -- (4.5,-2)
                           (4.5,-2.15) node{$y_{94}$}
                (4,-1)  -- (4,-2) 
                           (3.95,-2.15) node{$y_{93}$}
                (4,-1)  -- (3.5,-2) 
                           (3.4,-2.15) node{$y_{92}$}
                (4,-1)  -- (3,-2) 
                           (2.8,-2.15) node{$y_{91}$}
                
                (0,-1)  -- (1.5,-2) 
                           (1.7,-2.15) node{$y_{55}$}
                (0,-1)  -- (0.75,-2) 
                           (0.85,-2.15) node{$y_{54}$}
                (0,-1)  -- (0,-2) 
                            (0.3,-1.9) node{$y_{53}$}
                (0,-1)  -- (-0.75,-2) 
                            (-0.8,-2.15) node{$y_{52}$}
                (0,-1)  -- (-1.5,-2) 
                            (-1.6,-2.15) node{$y_{51}$}
                
                (-5,-2)  -- (-5.75,-3) circle (2pt) node[rectangle, draw=black!60, fill=black!5, very thick, minimum size=1mm]{}
                            (-6,-3.3)  node{$z_{111}$}
                (-5,-2)  -- (-5.25,-3) circle (2pt) node[rectangle, draw=black!60, fill=black!5, very thick, minimum size=1mm]{}
                            (-5.3,-3.3) node{$z_{112}$}
                (-5,-2)  -- (-4.75,-3) circle (2pt) node[rectangle, draw=black!60, fill=black!5, very thick, minimum size=1mm]{}
                            (-4.65,-3.3) node{$z_{113}$}
                (-5,-2)  -- (-4.25,-3) circle (2pt) node[rectangle, draw=black!60, fill=black!5, very thick, minimum size=1mm]{}
                            (-3.9,-3.3) node{$z_{114}$}
                
                (5,-2)  -- (5.75,-3) circle (2pt) node[rectangle, draw=black!60, fill=black!5, very thick, minimum size=1mm]{}
                            (6.1,-3.3) node{$z_{954}$}
                (5,-2)  -- (5.25,-3) circle (2pt) node[rectangle, draw=black!60, fill=black!5, very thick, minimum size=1mm]{}
                            (5.3,-3.3) node{$z_{953}$}
                (5,-2)  -- (4.75,-3) circle (2pt) node[rectangle, draw=black!60, fill=black!5, very thick, minimum size=1mm]{}
                            (4.6,-3.3) node{$z_{952}$}
                (5,-2)  -- (4.25,-3) circle (2pt) node[rectangle, draw=black!60, fill=black!5, very thick, minimum size=1mm]{}
                            (3.9,-3.3) node{$z_{951}$}
                
                (0,-2)  -- (0.75,-3) circle (2pt) node[rectangle, draw=black!60, fill=black!5, very thick, minimum size=1mm]{}
                            (1,-3.3) node{$z_{554}$}
                (0,-2)  -- (0.25,-3) circle (2pt) node[rectangle, draw=black!60, fill=black!5, very thick, minimum size=1mm]{}
                            (0.3,-3.3)node{$z_{553}$}
                (0,-2)  -- (-0.75,-3) circle (2pt) node[rectangle, draw=black!60, fill=black!5, very thick, minimum size=1mm]{}
                            (-1.1,-3.3)node{$z_{552}$}
                (0,-2)  -- (-.25,-3) circle (2pt) node[rectangle, draw=black!60, fill=black!5, very thick, minimum size=1mm]{}
                            (-.35,-3.3)node{$z_{551}$};
                
            
                
            
                             
                             
        \end{tikzpicture}

        \endpgfgraphicnamed
        
      \caption{The CA-Node-Steiner-Tree instance $T$ with its first four layers of Steiner nodes; the set of terminals is absent while final Steiner nodes are depicted with squares. Since $T$ is a tree, the solution to the CA-Node-Steiner-Tree problem is trivially $T$ itself.}
\label{fig:tree-define}
\end{figure}

%% file: fig3.tex
\begin{figure}[ht]
\centering
\scalebox{0.90}{

      \begin{tabular}{c|c}
           (a)
           \begin{tikzpicture}[scale=0.6]
            \draw[very thick, color=red]
                (0,1)  -- (-5,-1) 
                       -- (-5.75,-2) 
                (-2.5,-1)  -- (-3.25,-2) 
                (0,-1)  -- (-0.75,-2) 
                (2.5,-1)  -- (1.75,-2) 
                (5,-1)  -- (4.25,-2) ;
                
            \draw[very thick]
                
                (0,1) node[anchor= south east]{$4$+$b_i$} -- (-2.5,-1) node[anchor= south east]{$4$}
                                 (-4.5,-1.08) node[anchor= south east]{$7$+$b_i$}
                (0,1)  -- (0,-1) node[anchor= south west]{$4$}
                (0,1)  -- (2.5,-1) node[anchor= south west]{$4$}
                (0,1)  -- (5,-1) node[anchor= south west]{$4$}

                (-5.75,-2) circle (2pt) node[rectangle, draw=black!60, fill=black!5, very thick, minimum size=1mm]{}  
                    (-6.1,-2.54) node{$8$+$b_i$}
                
                (-5,-1)  -- (-5.25,-2) circle (2pt) node[rectangle, draw=black!60, fill=black!5, very thick, minimum size=1mm]{} 
                            (-5.25,-2.1)node[anchor= north]{$2$}
                
                (-5,-1)  -- (-4.75,-2) circle (2pt) node[rectangle, draw=black!60, fill=black!5, very thick, minimum size=1mm]{} 
                            (-4.75,-2.1) node[anchor= north]{$2$}
                
                (-5,-1)  -- (-4.25,-2) circle (2pt) node[rectangle, draw=black!60, fill=black!5, very thick, minimum size=1mm]{} 
                            (-4.25,-2.1) node[anchor= north]{$2$}

                (-3.25,-2) circle (2pt) node[rectangle, draw=black!60, fill=black!5, very thick, minimum size=1mm]{} 
                            (-3.25,-2.1) node[anchor= north]{$5$}
                                    
                (-2.5,-1)  -- (-2.75,-2) circle (2pt) node[rectangle, draw=black!60, fill=black!5, very thick, minimum size=1mm]{} 
                            (-2.75,-2.1) node[anchor= north]{$2$}
                                    
                (-2.5,-1)  -- (-2.25,-2) circle (2pt) node[rectangle, draw=black!60, fill=black!5, very thick, minimum size=1mm]{} 
                            (-2.25,-2.1) node[anchor= north]{$2$}
                                    
                (-2.5,-1)  -- (-1.75,-2) circle (2pt) node[rectangle, draw=black!60, fill=black!5, very thick, minimum size=1mm]{} 
                            (-1.75,-2.1) node[anchor= north]{$2$}

                (0,-1)  -- (0.75,-2) circle (2pt) node[rectangle, draw=black!60, fill=black!5, very thick, minimum size=1mm]{} 
                            (0.75,-2.1) node[anchor= north]{$2$}
                            
                (0,-1)  -- (0.25,-2) circle (2pt) node[rectangle, draw=black!60, fill=black!5, very thick, minimum size=1mm]{} 
                            (0.25,-2.1) node[anchor= north]{$2$}
                            
                (0,-1)  -- (-0.25,-2) circle (2pt) node[rectangle, draw=black!60, fill=black!5, very thick, minimum size=1mm]{} 
                            (-0.25,-2.1) node[anchor= north]{$2$}
                (-0.75,-2) circle (2pt) node[rectangle, draw=black!60, fill=black!5, very thick, minimum size=1mm]{} 
                            (-0.75,-2.1) node[anchor= north]{$5$}
                
                (2.5,-1)  -- (3.25,-2) circle (2pt) node[rectangle, draw=black!60, fill=black!5, very thick, minimum size=1mm]{} 
                            (3.25,-2.1) node[anchor= north]{$2$}
                            
                (2.5,-1)  -- (2.75,-2) circle (2pt) node[rectangle, draw=black!60, fill=black!5, very thick, minimum size=1mm]{} 
                            (2.75,-2.1) node[anchor= north]{$2$}
                            
                (2.5,-1)  -- (2.25,-2) circle (2pt) node[rectangle, draw=black!60, fill=black!5, very thick, minimum size=1mm]{} 
                            (2.25,-2.1) node[anchor= north]{$2$}
                            
                (1.75,-2) circle (2pt) node[rectangle, draw=black!60, fill=black!5, very thick, minimum size=1mm]{} 
                            (1.75,-2.1) node[anchor= north]{$5$}
                
                (5,-1)  -- (5.75,-2) circle (2pt) node[rectangle, draw=black!60, fill=black!5, very thick, minimum size=1mm]{} 
                        (5.75,-2.1) node[anchor= north]{$2$} 
                        
                (5,-1)  -- (5.25,-2) circle (2pt) node[rectangle, draw=black!60, fill=black!5, very thick, minimum size=1mm]{} 
                            (5.25,-2.1) node[anchor= north]{$2$}
                            
                (5,-1)  -- (4.75,-2) circle (2pt) node[rectangle, draw=black!60, fill=black!5, very thick, minimum size=1mm]{} 
                            (4.75,-2.1) node[anchor= north]{$2$}
                            
                (4.25,-2) circle (2pt) node[rectangle, draw=black!60, fill=black!5, very thick, minimum size=1mm]{} 
                            (4.25,-2.1) node[anchor= north]{$5$};

            \draw[dashed, very thick]
                (0,1)  -- (2,2.5) node[anchor=west]{$r$} ;

        \end{tikzpicture} &  
           (b)
           \begin{tikzpicture}[scale=0.6]
                \draw[very thick]
                    (0,1)circle (2pt) node[rectangle, draw=black!60, fill=black!5, very thick, minimum size=1mm]{}  
                        -- (0,0) 
                        -- (0,-1)
                        
                        (.5,1.15) node{$2$}
                        (.5,0.15) node{$4$}
                        (.5,-0.845) node{$5$}
                        
                    (0,0) -- (-1.75,-1) node[rectangle, draw=black!60, fill=black!5, very thick, minimum size=1mm]{}  
                    (0,0) -- (-2.75,-1) node[rectangle, draw=black!60, fill=black!5, very thick, minimum size=1mm]{}  
                             
                    (-1.75,-1.5) node{$2$}
                    (-2.75,-1.5) node{$2$}
                    (-3.9,-1.5) node{$5$}
                    
                    (0,-1) -- (4.5,-3) 
                              (4.5,-2.85)node[anchor=west]{$4$}
                    (0,-1) -- (1.7,-3) 
                              (1.7,-2.85) node[anchor=west]{$4$}
                    (0,-1) -- (-1.7,-3) 
                              (-1.7,-2.85) node[anchor=east]{$4$}
                    (0,-1) -- (-4.5,-3) 
                              (-4.5,-2.85) node[anchor=east]{$8$}
                    (0,-1) -- (7,-3) 
                              (7,-2.9) node[anchor=west]{$r$}
                    
                    (-4.5,-3) -- (-5.5,-4)
                    (-4.5,-3) -- (-5.5 + 11/16,-4)circle (2pt) node[rectangle, draw=black!60, fill=black!5, very thick, minimum size=1mm]{}  
                                 (-5.5 + 11/16,-4.2)node[anchor= north]{$2$}
                    (-4.5,-3) -- (-5.5 + 22/16,-4)circle (2pt) node[rectangle, draw=black!60, fill=black!5, very thick, minimum size=1mm]{}  
                                 (-5.5 + 22/16,-4.2)node[anchor= north]{$2$}
                    (-4.5,-3) -- (-5.5 + 33/16,-4)circle (2pt) node[rectangle, draw=black!60, fill=black!5, very thick, minimum size=1mm]{}  
                                 (-5.5 + 33/16,-4.2)node[anchor= north]{$2$}
                    
                    (-1.7,-3) -- (-5.5+44/16,-4)
                    (-1.7,-3) -- (-5.5+55/16,-4)circle (2pt) node[rectangle, draw=black!60, fill=black!5, very thick, minimum size=1mm]{}  
                                (-5.5+55/16,-4.2)node[anchor= north]{$2$}
                    (-1.7,-3) -- (-5.5 + 66/16,-4)circle (2pt) node[rectangle, draw=black!60, fill=black!5, very thick, minimum size=1mm]{}  
                                (-5.5 + 66/16,-4.2)node[anchor= north]{$2$}
                    (-1.7,-3) -- (-5.5 + 77/16,-4)circle (2pt) node[rectangle, draw=black!60, fill=black!5, very thick, minimum size=1mm]{}  
                                (-5.5 + 77/16,-4.2) node[anchor= north]{$2$}
                    
                    (1.7,-3) -- (5.5-77/16,-4)
                    (1.7,-3) -- (5.5-66/16,-4)circle (2pt) node[rectangle, draw=black!60, fill=black!5, very thick, minimum size=1mm]{}  
                                (5.5-66/16,-4.2)node[anchor= north]{$2$}
                    (1.7,-3) -- (5.5 - 55/16,-4)circle (2pt) node[rectangle, draw=black!60, fill=black!5, very thick, minimum size=1mm]{}  
                                (5.5 - 55/16,-4.2)node[anchor= north]{$2$}
                    (1.7,-3) -- (5.5 - 44/16,-4)circle (2pt) node[rectangle, draw=black!60, fill=black!5, very thick, minimum size=1mm]{}  
                                (5.5 - 44/16,-4.2)node[anchor= north]{$2$}
                    
                    (4.5,-3) -- (5.5 - 33/16,-4)
                    (4.5,-3) -- (5.5-22/16,-4)circle (2pt) node[rectangle, draw=black!60, fill=black!5, very thick, minimum size=1mm]{}  
                                (5.5-22/16,-4.2)node[anchor= north]{$2$}
                    (4.5,-3) -- (5.5 - 11/16,-4)circle (2pt) node[rectangle, draw=black!60, fill=black!5, very thick, minimum size=1mm]{}  
                                (5.5 - 11/16,-4.2)node[anchor= north]{$2$}
                    (4.5,-3) -- (5.5,-4)circle (2pt) node[rectangle, draw=black!60, fill=black!5, very thick, minimum size=1mm]{}  
                                (5.5,-4.2)node[anchor= north]{$2$};
                    
                \filldraw[color=gray!50,dashed]
                    (7,-3) .. controls (6,-5) .. (7,-5)
                    (7,-3) .. controls (8,-5) .. (7,-5);
                \draw[very thick, dashed]
                    (7,-3) .. controls (6,-5) .. (7,-5)
                    (7,-3) .. controls (8,-5) .. (7,-5);
                \draw[color=red,very thick]
                    (0,0) -- (-3.9,-1) node[rectangle, draw=black!60, fill=black!5, very thick, minimum size=1mm]{}
                    (0,-1) -- (-4.5,-3) -- (-5.5,-4)
                    (-1.7,-3) -- (-5.5+44/16,-4)
                    (-4.5,-3) -- (-5.5,-4)
                    (1.7,-3) -- (5.5-77/16,-4)
                    (4.5,-3) -- (5.5 - 33/16,-4);
                
                \draw

                    (-5.5,-4)circle (2pt) node[rectangle, draw=black!60, fill=black!5, very thick, minimum size=1mm]{}      
                        (-5.5,-4.21)node[anchor= north]{$9$}
                    
                    (-5.5+44/16,-4)circle (2pt) node[rectangle, draw=black!60, fill=black!5, very thick, minimum size=1mm]{}  
                        (-5.5+44/16,-4.21)node[anchor= north]{$5$}
                    
                    (5.5-77/16,-4)circle (2pt) node[rectangle, draw=black!60, fill=black!5, very thick, minimum size=1mm]{}  
                        (5.5-77/16,-4.21)node[anchor= north]{$5$}
                    
                    (5.5 - 33/16,-4)circle (2pt) node[rectangle, draw=black!60, fill=black!5, very thick, minimum size=1mm]{}  
                        (5.5 - 33/16,-4.21)node[anchor= north]{$5$};

      \end{tikzpicture}
  \end{tabular}
 }

        \endpgfgraphicnamed

      \caption{(a) The $w$-values for the nodes of $Q_{x_i}$ are described pictorially. The term $b_i$ is equal to $8$ if $i=1$ and equal to $1$ otherwise. The red edges indicate the paths found by Algorithm~\ref{alg:computing_w}. (b) A clearer explanation of the $w$-values for every node in $Q_{x_2}$ is given here, as well as the shape of $Q_{x_2}$ when $T$ is rooted at the node $z_{222}$.}
\label{fig:bad-example}
\end{figure}

%% file: appendix.tex
\section{Missing proofs}
\label{appendix:proofs}

\subsection{Proof of Theorem~\ref{cor:approx-preserve}}\label{appendix:approx-preserve}

Let $(G=(V,E),L)$ be a given CacAP instance, where $L \subseteq \binom{V}{2}$. For every link $\ell = (v_0,v_{n+1})$ let $v_1,\dots,v_n$ be the sequence of nodes of degree at least $4$ that belong to every simple path from $v_0$ to $v_{n+1}$ in $G$, such that each pair $\ell_i = \{v_i, v_{i+1}\}$ lies on the same cycle $C_i$ in $G$. We call each $\ell_i$ a \emph{projection} of $\ell$ on $C_i$, and the set of projections of $\ell$ is denoted by $\mathrm{proj}(\ell)$. For projections $\ell_i = (v, u)$ and $\ell_i' = (v',u')$ of $\ell \in L$ and $\ell' \in L$ respectively, with endpoints on the same cycle $C_i$, then $\ell_i$ and $\ell_i'$ are said to \emph{cross} if they satisfy one of the following conditions: they share an endpoint, or one of the two simple $v$ to $u$ paths in $C_i$ contain exactly one of $u'$ or $v'$. Links $\ell$ and $\ell'$ are said to \emph{cross} if they have projections $\ell_i\in \mathrm{proj}(\ell)$ and $\ell_i'\in \mathrm{proj}(\ell')$ that cross. 

From $(G,L)$ we construct a Node Steiner Tree instance $G_{ST} = (R\cup L, E_{ST})$ in the following way. The set of links $L$ make up the Steiner nodes of $G_{ST}$ and the set of degree-$2$ nodes in $G$ make up the set of terminals $R$. For each link $\ell \in L$ with endpoint $r\in R$, we have $(\ell, r)\in E_{ST}$, and for each pair of links $\ell$ and $\ell'$ that cross, we have $(\ell,\ell') \in E_{ST}$. Notice that $G_{ST}$ is, in particular, a CA-Node-Steiner-Tree instance.

A key lemma that shows why the above reduction is useful is the following; its proof is given in~\cite{DBLP:conf/icalp/BasavarajuFGMRS14,DBLP:conf/stoc/Byrka0A20}.

\begin{lemma}[\cite{DBLP:conf/icalp/BasavarajuFGMRS14,DBLP:conf/stoc/Byrka0A20}]
    Let $(G, L)$ be a CacAP instance, and let $G_{ST} = (R\cup L, E_{ST})$ be the corresponding CA-Node-Steiner-Tree instance. A subset of links $L'\subseteq L$ is a feasible solution of $(G,L)$ if and only if $G_{ST}[R\cup L']$ is connected.
\end{lemma}

Clearly, a feasible solution $L'$ to a CacAP instance $(G,L)$ has size $|L'|$ and implies a feasible solution to the corresponding CA-Node-Steiner-Tree instance $G_{ST} = (T\cup L, E_{ST})$ of size $|L'|$ and vice versa, so we can conclude Theorem~\ref{cor:approx-preserve}.


\input{./alternative.tex}

\subsection{Proof of Theorem~\ref{thm:k-restricted-decomp}}\label{appendix:k-restricted-decomp}

Before proving Theorem~\ref{thm:k-restricted-decomp}, we start with the following observations and lemmas.

\begin{observation}\label{obs:lower-bound}
    Let $\OPT$ be the optimal cost of a CA-Node-Steiner-Tree instance with $t$ terminals. Then, we have $\OPT \geq t/2$.
\end{observation}

\begin{definition}
    A rooted binary tree is called regular if every non-leaf node has exactly $2$ children.
\end{definition}

As a clarification, a leaf in this rooted tree is a vertex with no children. For the rest of this section, we use the notation $P_T(u,v)$ to denote the unique path from $u$ to $v$ in a tree $T$.

\begin{lemma}[see also Lemma 3.1 in~\cite{DBLP:journals/siamcomp/BorchersD97}]\label{lemma:map}
    For any rooted regular binary tree $T = (V, E)$ with a set of leaves $C \subseteq V$, there exists a one-to-one mapping $f: V \setminus C \to C$ from non-leaf nodes to leaves, such that
    \begin{itemize}
        \item for every $u \in V \setminus C$, $f(u)$ is a descendent of $u$, and
        \item all paths $P_T(u,f(u))$ from $u \in V \setminus C$ to $f(u)$ are pairwise edge disjoint, and internally node disjoint (where internally node disjoint means that there is no node that belongs to both paths and is also internal for both paths).
    \end{itemize}
\end{lemma}

\begin{proof}
We will prove a slightly stronger statement which implies the lemma. We will show that for any regular binary tree $T = (V, E)$ rooted at a vertex $r \in V$ with a set of leaves $C \subseteq V$, there exists a one-to-one mapping $f: V \setminus C \to C$ from non-leaves to leaves such that:
\begin{enumerate}
    \item for every $u \in V \setminus C$, $f(u)$ is a descendant of $u$,
    \item all paths $P_T(u,f(u))$ from $u$ to $f(u)$ are pairwise edge disjoint, and internally node disjoint,
    \item there exists a leaf $c_0 \in C$ that is not the image of any non-leaf node, such that the path $P_T(r, c_0)$ from the root $r$ of the tree to $c_0$ is edge disjoint and internally node disjoint from any other path $P_T(u, f(u))$, for any $u \in V \setminus C$.
\end{enumerate}

To prove the above statement, we use induction on the height $h$ of the tree, where we define the height of the tree as the maximum number of nodes that appear in any path from the root to a leaf, including the endpoints. The base case is $h = 1$, in which case, the tree consists of a single leaf, and thus the statement is trivially true. Suppose now that the statement holds for all regular binary trees of height at most $h \geq 1$. We will show that the statement also holds for all regular binary trees of height $h+1$. 

Let $T = (V, E)$ be a tree of height $h+1$, let $r$ be its root, and let $u_1$ and $u_2$ be its two children. Let $C \subseteq V$ be its set of leaves. The subtrees $T_1$ and $T_2$, rooted at $u_1$ and $u_2$, respectively, are both regular binary trees of height at most $h$. Let $V_i$ denote the set of vertices of $T_i$, for $i \in \{1,2\}$, and $C_i \subseteq V_i$ denote the set of leaves of $T_i$, for $i \in \{1,2\}$. Note that $V = \{r\} \cup V_1 \cup V_2$ and $C = C_1 \cup C_2$. By the induction hypothesis, there exist maps $f_1: V_1 \setminus C_1 \to C_1$ for $T_1$ and $f_2: V_2 \setminus C_2 \to C_2$ for $T_2$ that satisfy the desired properties. Let $c_i \in C_i$ be the leaf of $T_i$ that satisfies the third property of the statement, for $i \in \{1,2\}$. We now define the map $f: V \setminus C \to C$ for $T$ as follows:
\begin{itemize}
    \item $f(r) = c_1$,
    \item for every $i \in \{1,2\}$ and $u \in V_i \setminus C_i$, we set $f(u) = f_i(u)$,
\end{itemize}

Clearly, the first property of the statement is satisfied. By the induction hypothesis, and since $T_1$ and $T_2$ are disjoint, we get that all paths $\{P_T(u,f(u))\}_{u \in V \setminus (C \cup \{r\})}$ are pairwise edge disjoint, and internally they are node disjoint. Thus, we only need to check the path $P_T(r,c_1))$. We observe that the path $P_T(r,c_1)$ from $r$ to $c_1$ is clearly node disjoint from all paths of the set $\{P_T(u,f(u))\}_{u \in V_2 \setminus C_2}$ and by the induction hypothesis, it is also edge disjoint and internally node disjoint from any path $P_T(u,f(u))$ where $u \in V_1 \setminus C_1$. Again, we clarify here that internally node disjoint means that there is no node that belongs to both paths and is also internal for both paths. Thus, the first two properties are satisfied. In order to prove the third property, we set $c_0 \equiv c_2$. By the induction hypothesis, the path $P_T(v_2, c_2)$ is edge disjoint and internally node disjoint from any other path of the set $\{P_T(u,f(u))\}_{u \in V_2 \setminus C_2}$. Moreover, it only intersects the path from $r$ to $f(r)$ at its endpoint. Thus, the third property is satisfied for $T$. This concludes the proof.
\end{proof}

\begin{lemma}\label{lemma:high-degree-vertices-vs-leaves}
Let $T = (V, E)$ be a tree with $t$ leaves, and let $d(v)$ be the degree of $v \in V$. Then, we have $\sum_{v\in V : d(v) \geq 3} (d(v) - 2) \leq t$.
\end{lemma}

\begin{proof}
Let $n$ be the total number of vertices of $T$, which we decompose as $n = t + n_2 + n_{\geq 3}$, where $n_2$ is the number of vertices of degree exactly 2, and $n_{\geq 3}$ is the number of vertices of degree at least $3$. We have $t + 2n_2 + \sum_{v \in V: d(v) \geq 3} d(v) = 2(n-1) = 2t + 2n_2 + 2n_{\geq 3} - 2$, which gives $\sum_{v \in V: d(v) \geq 3} (d(v) - 2) = t - 2 \leq t$.
\end{proof}

We are now ready to prove Theorem~\ref{thm:k-restricted-decomp}.

\begin{proof}[Proof of Theorem~\ref{thm:k-restricted-decomp}]
    We begin with an optimal solution $Q$ of cost $\OPT$ to the given CA-Node-Steiner-Tree instance, whose set of leaves is identical to the set of terminals. We modify $Q$ as described in Xu et al.~\cite{xu2010approximations} to get a new tree $Q^+$ of equal cost. More precisely, we first create a dummy Steiner node of zero weight in the middle of an edge that serves as the root of the tree; thus, from now on $Q$ will be a rooted tree. For each Steiner node $v$ that has at least three children, we expand $v$ into a binary tree by using duplicate nodes of zero weight such that every node in the resulting tree has at most two children. The following lemma, whose proof is straightforward and thus omitted, shows that the number of auxiliary nodes is proportional to the degree of the node that is modified.
    
    \begin{lemma}\label{lemma:degree}
        Let $v$ be a node of $Q$ with at least three children. Then, if we apply the transformation described above to $v$, we get a new tree that contains $d(v)-3$ extra auxiliary nodes, where $d(v)$ is the degree of $v$ in $Q$.
    \end{lemma}%

    We repeatedly apply the transformation described above, and we get a new tree $Q^+$ where every Steiner node has at most two children. This new tree $Q^+$ has at most $n + \sum_{v \in Q: d(v) \geq 4}(d(v) - 3) \leq 3n$ vertices, where $n$ is the number of vertices of the rooted tree $Q$. Finally, for any non-leaf node different from the root that has degree $2$, we add a dummy terminal as a child, and we end up with a regular node weighted binary tree (where only Steiner nodes have weight), which, by slightly abusing notation, we call $Q^+$. 
    
    We now label the Steiner nodes of $Q^+$ with labels $\{0,\dots, m-1\}$ as follows; we stress that the labeling and the arguments used closely follow the techniques of Borchers and Du~\cite{DBLP:journals/siamcomp/BorchersD97}. We first clarify some terminology. We say that the root is at the $0^{\mathrm{th}}$ layer of $Q^+$, the children of the root are at the $1^{\mathrm{st}}$ layer, and so on. We label the root with $0$, the $1^{\mathrm{st}}$ layer of nodes with $1$, and so on until the $(m-1)^{\mathrm{th}}$ layer. We then repeat the labelling with label $0$, starting from layer $m$. More precisely, the Steiner nodes at layer $i$ all get the same label $i\mod m$. We obviously have the property that any $m$ consecutive layers have different labels.
    
    For each label $j \in \{0, \ldots, m-1\}$, we will define a feasible $k$-restricted Steiner tree $Q_j^+$. We first define its components; the tree $Q_j^+$ is simply the union of these components and its cost is the sum of the costs of its components. For a fixed label $j$, each component of $Q_j^+$ is rooted at either the root of $Q^+$ or a Steiner node with label $j$. We define the component rooted at Steiner node $v$ labelled with $j$ as follows: the component rooted at $v$ contains the paths in $Q^+$ to the first nodes below $v$ that are also labelled with $j$; call these nodes \emph{intermediate leaves}. More precisely, the component rooted at $v$ connects to $u$ if $u$ is a descendant of $v$ labelled with $j$, such that there is no other node of label $j$ on the path from $v$ to $u$ besides $v$ and $u$. Then, for an intermediate leaf $u$ that is a Steiner node, the component contains the unique path $p(u)$, from $u$ to $f(u)$, as defined in Lemma~\ref{lemma:map}. If there is a path from $v$ to a leaf of $Q^+$ that is a descendant of $v$ such that the path does not contain a node with label $j$, then that path is also in the component. The component of $Q_j^+$ that is rooted at the root of $Q^+$ is constructed in a similar way. Finally, $Q_j^+$ is the union of all components rooted at a Steiner node with label $j$, along with the component rooted at the root of $Q^+$, which is always included, even if $j > 0$. 
    
    \begin{lemma}
        Each component of $Q_j^+$ has at most $k$ terminals.
    \end{lemma}
    \begin{proof}
        It is clear that a component has at most $2^m$ intermediate leaves, and thus at most $2^m$ leaves connected by paths from these intermediate leaves. Moreover, if there is a direct path from the root of the component to a leaf that does not contain an intermediate leaf, then the leaf must be no more than $m$ layers below the component root. So, the total number of leaves for any component is at most $2^m = k$.
    \end{proof}

    \begin{lemma}\label{lemma:feasible}
        Let $j \in \{0, \ldots, m-1\}$. The tree $Q_j^+$ is a feasible $k$-restricted CA-Node-Steiner-Tree.
    \end{lemma}
    \begin{proof}
        To see this, we only need to look at the part of each component that connects a node $v$ labelled as $j$ to its intermediate leaves. It is easy to see that the union of all these components, along with the component rooted at the root of $Q^+$, contains all nodes of $Q^+$, and so $Q_j^+$ is a feasible $k$-restricted Steiner tree.
    \end{proof}

    Given a $k$-restricted Steiner tree $Q_j^+$ of $Q^+$, we now explain how to construct a $k$-restricted Steiner tree $Q_j$ of $Q$. For any component $B$ of $Q_j^+$ we construct a component $B'$ in $Q$ by contracting every expanded node and removing every dummy terminal. Moreover, if $B'$ contains the global root, we remove it as well and ``restore" the edge that was split in order to create the root. $B'$ clearly still has at most $k$ terminals, and thus the union of these components of $Q$ is a $k$-restricted Node Steiner Tree, denoted as $Q_j$. For the final step of our argument, we require the following observation about $k$-restricted trees that will be useful in bounding the optimal cost.

    \begin{observation}\label{obs1}
        Consider the $k$-restricted Steiner trees $Q_0^+, \ldots, Q_{m-1}^+$ of $Q^+$, as defined above, and their corresponding components. Any non-root Steiner node appears as an intermediate leaf in exactly one component of exactly one such Steiner tree.
    \end{observation}

    Using this observation we will now find a bound on the sum $\sum_{j = 0}^{m-1}cost(Q_j)$. For that, we consider a Steiner node $v$ of $Q$. We first note that $v$ has $d(v) - 2$ ``copies'' in $Q^+$. By Observation~\ref{obs1}, each such copy appears as an intermediate leaf in exactly one component of exactly one Steiner tree among $Q_0^+, \ldots, Q_{m-1}^+$, and thus it is counted twice, since it also appears once as the root of a component. Finally, each such copy can appear at most one more time as an internal node of a path from an intermediate leaf to a leaf (as implied by Lemma~\ref{lemma:map}). Thus, $v$ appears at most $2(d(v) - 2) + m$ times in the sum $\sum_{j = 0}^{m-1}cost(Q_j)$. Let $S$ denote the set of Steiner nodes of $Q$, $S_2 = \{s \in S: d(s) = 2\}$ and $S_{\geq 3} = \{ s \in S: d(s) \geq 3\}$. Note that $S = S_2 \cup S_{\geq 3}$ and $\OPT = |S|$. Let $t$ be the number of terminals. We have
    \begin{align*}
    \sum_{j = 0}^{m-1}cost(Q_j) &\leq m|S| + 2\sum_{s \in S} (d(s)-2) = m|S| - 4|S|+ 2\sum_{s \in S} d(s) = (m-4)|S| + 2\sum_{s \in S_2} d(s) + 2\sum_{s \in S_{\geq 3}} d(s)\\
            &\leq (m-4)|S| + 4 \left(|S| - |S_{\geq 3}| \right)  + 2\sum_{s \in S_{\geq 3}} d(s) = m \cdot \OPT + 2\sum_{s \in S_{\geq 3}} (d(s) - 2) \leq m \cdot \OPT + 2t\\
            &\leq m \cdot \OPT + 4\OPT = (m+4)\OPT,
    \end{align*}
    where in the above derivations we used Observation~\ref{obs:lower-bound} and Lemma~\ref{lemma:high-degree-vertices-vs-leaves}. Thus, there exists a label $j \in \{0, 1, \ldots, m-1\}$ such that $cost(Q_j) \leq \left(1 + \frac{4}{\log k} \right) \OPT$, where we recall that $k = 2^m$.
\end{proof}

\subsection{Proof of Theorem~\ref{thm:eps_approx}}\label{appendix:proof-eps_approx}

\begin{proof}[Proof of Theorem~\ref{thm:eps_approx}]
Given an $\varepsilon > 0$, we set $k = 2^{\lceil 4/\varepsilon \rceil}$. Since
$k$ is fixed, computing a minimum Node Steiner tree connecting $k$ terminals can be done in polynomial time. To see this, we note that the Node Steiner Tree problem can be reduced to the Directed Edge Steiner Tree problem without modifying the number $k$ of terminals (see~\cite{DBLP:journals/networks/Segev87}). Thus, one can use existing algorithms for optimally solving Directed Edge Steiner Tree instances with a fixed number of terminals (see, e.g.,~\cite{DBLP:conf/focs/FeldmanR99}). We conclude that the set of components $\mathbf{C}$ can be computed in polynomial time, and thus we can generate the variables of $k$-DCR LP in polynomial time. Although the set of constraints is exponential in $|R|$, an optimal solution for it can be computed in polynomial time by standard flow techniques (see~[Lemma 8, \cite{DBLP:conf/stoc/Byrka0A20}]).

Regarding the LP value, it is easy to see that $\OPT_{\LP}(|R|) \leq \OPT$, as the entire optimal tree can be viewed as an $|R|$-restricted Steiner Tree which can be directed towards $r$, and  whose cost is equal to $\OPT$.  We now argue that $\OPT_{\LP}(k)$ is at most $\left(1 + \frac{4}{\log k} \right)\OPT_{\LP}(|R|)$. For this, let $x$ be an arbitrary feasible solution of the $|R|$-DCR LP. We look at every component $Q$ with terminals $R'$ and sink $c' \in R'$ whose $x$-value is strictly positive, i.e., $x(Q) > 0$, and we do the following: if the component is a $k$-restricted component, we set $x'(Q) =x(Q)$. Otherwise, by Theorem~\ref{thm:k-restricted-decomp}, we compute a $k$-restricted Steiner Tree $Q'$ with components $Q_1', \ldots, Q_j'$, each containing at most $k$ terminals, such that $cost(Q') = \sum_{i = 1}^j cost(Q_i') \leq \left(1 + \frac{4}{\log k} \right) cost(Q)$. Moreover, we ``direct'' the components consistently towards the sink $c'$ of $Q$, and increase the corresponding $x'$-variables by $x(Q')$ for each such directed component. The resulting vector $x'$ is a feasible solution for $k$-DCR, whose objective function value is at most $\left(1 + \frac{4}{\log k} \right)$ times the objective function value of $x$. We therefore get $\OPT_{\LP}(k) \leq \left(1 + \frac{4}{\log k} \right) \OPT_{\LP}(|R|) \leq  (1 + \varepsilon) \OPT$.
\end{proof}

%% file: alternative.tex
\subsection{Proof of Theorem~\ref{cor:approx-preserve-block}}\label{appendix:approx-preserve-block}
In this section we slightly generalize the reduction from a Block-TAP instance to CA-Node-Steiner-Tree that was described by Nutov~\cite{nutov20202nodeconnectivity}, and give a reduction from 1-Node-CAP to CA-Node-Steiner-Tree that shows that an $\alpha$-approximation for CA-Node-Steiner-Tree implies an $\alpha$-approximation for 1-Node-CAP.

The reduction has two steps. We first reduce the original problem to a weighted Block-TAP instance with specific properties, where each link has weight either $0$ or $1$. We then show that Nutov's reduction~\cite{nutov20202nodeconnectivity} still applies, and it gives a final CA-Node-Steiner-Tree instance that is unweighted.

\subsubsection{The reduction to Block-TAP with $\{0,1\}$ weights}\label{appendix:1node-block-tap}

In this section, we use one of the standard ways to reduce a 1-Node-CAP instance to a Block-TAP instance. Let $G = (V, E)$ be a connected graph, and let $L \subseteq \binom{V}{2}$ be a set of links. Let $C \subseteq V$ be the set of cut nodes of $G$. As a reminder, a node $c \in V$ is a cut node of $G$ if $G \setminus \{c\}$ is not connected. Note that if $C = \emptyset$, then $G$ is already $2$-node-connected, and so the problem is trivial. So, from now on we assume that $C \neq \emptyset$.

We start by computing the block-cut tree $T = (V_T, E_T)$ of $G$. It is well-known that the block-cut tree of a graph can be computed in polynomial time. $T$ contains one copy of each cut node of $G$ and one node for each maximal 2-node connected component of $G$. Let $C = \{c_1, \ldots, c_n\} \subseteq V$ be the set of cut nodes of $G$, and let $\mathcal{B} = \{B_1, \ldots, B_m\}$ be the collection of maximal 2-node connected components. We have $V_T = C \cup \mathcal{B}$. The edge set $E_T$ is defined as follows: for every $i \in [n], j \in [m]$, $(c_i, B_j) \in E_T$ if and only if $c_i \in B_j$. It is easy to verify that the resulting graph $T$ is indeed a tree. 

Along with a block-cut tree $T$ of a graph $G$, we also define a function $f_G$ that maps the nodes of $G$ to the nodes of $T$. More precisely, we define a map $f_G: V\to V_T$ as follows: $f_G(c_i) = c_i$, for every $i \in [n]$, and $f_G(v) = B_j$ for every $j \in [m]$, $v \in B_j \setminus C$. Note that $f_G$ is well-defined, i.e., for every $v \in V \setminus C$, there exists exactly one block $B_j$ such that $v \in B_j$. We also extend the map $f_G$ and define $f_G(\ell) \coloneqq (f_G(u), f_G(v))$, for any $\ell = (u,v) \in L$. Finally, the set of links $L_T \subseteq \binom{V_T}{2}$ is defined as follows:
\begin{itemize}
    \item for each link $\ell = (u,v) \in L$ there is a link $\ell_T \coloneqq f_G(\ell) \in L_T$ of weight $1$. Let $f_G(L)$ denote the set of all such links.
    \item for each block $B_j$, $j \in [m]$, and every pair of cut vertices $c, c' \in B_j \cap C$, where $c \neq c'$, there is a link $(c,c') \in L_T$ of weight $0$. Let $L_0$ denote the set of all the links of this form.
\end{itemize}

Throughout the rest of this section, we will use the notation introduced in these last two paragraphs. The following observation is now easy to verify.

\begin{observation}\label{obs:block-cut}
Let $G = (V, E)$ be a connected graph and let $T = (V_T, E_T)$ be a block-cut tree of $G$. Let $f_G: V \to V_T$ be the associated map. Let $c$ be a cut vertex of $G$, and let $G_1, \ldots, G_k$ be the connected components of $G \setminus \{c\}$, where $k \in \mathbb{N}_{\geq 2}$. Let $T_1, \ldots, T_{k'}$ be the connected components of $T \setminus \{c\}$. Let $V_i \subseteq V$ be the set of nodes contained in the blocks that are nodes of $T_i$, for every $i \in [k']$. Then, the following hold:
\begin{enumerate}
    \item $k = k'$,
    \item there is a bijection $g_c: [k] \to [k]$, such that $V_{g_c(i)} = V(G_i) \cup \{c\}$ for every $i \in [k]$, where $V(G_i)$ is the set of nodes contained in $G_i$,
    \item $V_i \cap V_j = \{c\}$ for every $i \neq j \in [k]$.
\end{enumerate}
\end{observation}

We will now show that the two instances are equivalent. More precisely, we prove the following lemma.

\begin{lemma}\label{lemma:reduction-to-block-cut}
Let $G = (V, E)$ be a connected graph and let $L$ be a set of links in $G$. Let $T = (V_T, E_T)$ be the corresponding block-cut tree of $G$, let $f_G$ be the associated map between $V$ and $V_T$, and let $L_T = f_G(L) \cup L_0$ be the corresponding set of links. Let $L' \subseteq L(G)$. Then, $G \cup L'$ is 2-node-connected if and only if $T \cup f_G(L') \cup L_0$ is 2-node-connected.
\end{lemma}

Before proving the lemma, we note that we might have two (or more) links $\ell\neq \ell' \in L$ being mapped, via $f_G$, to the same $\bar{\ell}$ link in $T$; in other words, $f_G^{-1}(\bar{\ell})$ might be a set with more than one link. In such a case, we map back $\bar{\ell}$ to any link of $f_G^{-1}(\bar{\ell})$ in an arbitrary way.

\begin{proof}[Proof of Lemma~\ref{lemma:reduction-to-block-cut}]

Let $L' \subseteq L$ be a set of links such that $G \cup L'$ is $2$-node-connected. We will show that $T \cup f_G(L') \cup L_0$ is $2$-node-connected. Suppose otherwise. This means that there exists a node $u \in V_T$ that, if removed, disconnects the tree. We observe that $u$ must necessarily be a cut node, i.e., $u = c \in C$. To see this, note that in the graph $T \cup L_0$, all nodes that are adjacent to a node $B_j$, for any $j \in [m]$, form a clique, and thus, the removal of $B_j$ cannot disconnect the graph $T \cup L_0(T)$. Since $u = c$ is a cut node, it means that $G \setminus \{c\}$ is not connected. By Observation~\ref{obs:block-cut}, the connected components $G_1, \ldots, G_k$ of $G \setminus \{c\}$ correspond, one by one, to the components $T_1, \ldots, T_k$ of $T \setminus \{c\}$. This implies that any link $\ell \in L'$ connecting component $G_i$ with component $G_j$, $i \neq j \in [k]$, corresponds to a link $f_G(\ell)$ connecting $T_i$ with $T_j$. Thus, since $(G \cup L') \setminus \{c\}$ is connected, this means that $(T \cup f_G(L')) \setminus \{c\}$ is connected, and so we get a contradiction. We conclude that $T \cup f_G(L') \cup L_0(T)$ is $2$-node-connected.

Conversely, let $L' \subseteq L$ such that $T \cup f_G(L') \cup L_0$ is $2$-node-connected. We will show that $G \cup L'$ is $2$-node-connected. Suppose otherwise. This means that there exists a cut node $c \in C$ such that $(G \cup L') \setminus \{c\}$ is not connected. In particular, there are at least $2$ connected components $H_1$ and $H_2$ in $(G \cup L') \setminus \{c\}$, and these components are simply unions of the connected components $G_1, \ldots, G_k$ of $G \setminus \{c\}$. Since there are no links with $c$ as an endpoint in the graph $(T \cup f_G(L') \cup L_0) \setminus \{c\}$, this means that the connected components of $(T \cup f_G(L') \cup L_0) \setminus \{c\}$ are exactly the same as the connected components of $(T \cup f_G(L')) \setminus \{c\}$. Moreover, by Observation~\ref{obs:block-cut}, the components $G_1, \ldots, G_k$ correspond, one by one, to the components $T_1, \ldots ,T_k$ of $T \setminus \{c\}$. Since $(T \cup f_G(L')) \setminus \{c\}$ is connected, this means that the set $f_G(L')$ contains links that do not have $c$ as an endpoint which connect all of the components $T_1, \ldots, T_k$ to a single component; in particular, if $u \in H_1$ and $v \in H_2$, then there is a path between $f_G(u)$ and $f_G(v)$ in $(T \cup f_G(L')) \setminus \{c\}$. This means that $H_1$ and $H_2$ must be connected via $L'$, and so we get a contradiction. We conclude that $G \cup L'$ is $2$-node-connected.
\end{proof}

\subsubsection{The reduction to CA-Node-Steiner-Tree}

Given the Block-TAP instance $T = (V_T, E_T)$, as constructed in the previous section, we now modify Nutov's reduction~\cite{nutov20202nodeconnectivity} in order to reduce our weighted Block-TAP instance to an unweighted CA-Node-Steiner-Tree instance. To simplify notation, we denote the set of links in $T$ as $L = L_0 \cup L_1$, where $L_0$ is the set of links of weight $0$ and $L_1$ is the set of links of weight $1$. Let $R \subseteq E_T$ be the set of edges of $T$ that have a leaf as an endpoint; these edges will correspond to terminals in the resulting CA-Node-Steiner-Tree instance. We define three graphs, with the last one being the final resulting CA-Node-Steiner-Tree instance.

\begin{definition}
\phantom{}
\begin{enumerate}
    \item The \emph{$(L,E_T)$-incidence graph} is the graph with node set $L \cup E_T$ and edge set defined as follows: there is an edge between $\ell = (u,v) \in L$ and $e\in E_T$ if $e$ belongs to the unique path in $T$ connecting $u$ and $v$.
        
    \item The \emph{short-cut $(L, E_T)$-incidence graph} is obtained from the $(L, E_T)$-incidence graph by short-cutting the set $E_T$, where short-cutting a node $e \in E_T$ means adding a clique between the neighbors of $e$.
        
    \item The \emph{reduced $(L,E_T)$-incidence graph} is obtained from the short-cut $(L,E_T)$-incidence graph as follows. For every $\ell \neq \ell' \in L_1$ such that there is a path between $\ell$ and $\ell'$ in the short-cut $(L,E_T)$-incidence graph whose nodes, other than the endpoints, all belong to $L_0$, we add an edge between $\ell$ and $\ell'$.We then delete the set $E_T \setminus R$ and the set $L_0$. In this resulting instance, let $R$ be the set of terminals and $L_1$ be the set of Steiner nodes.
    \end{enumerate}
\end{definition}
We stress that the final reduced $(L,E_T)$-incidence graph is treated as an unweighted Node Steiner Tree instance, i.e, all Steiner nodes have weight $1$.

For $J = T \cup L$, we define $\kappa_J(s,t)$ to be the maximum number of internally node disjoint paths from $s$ to $t$ in $J$. Given a block-cut tree $T = (V_T, E_T)$ with a set of links $L = L_0 \cup L_1$, let $H$ be the $(L,E_T)$-incidence graph, and let $H_R$ be the reduced $(L,E_T)$-incidence graph. For distinct vertices $s,t\in V_T$, we say $(s,t)$ are \emph{$H$-reachable} if $H$ has a path connecting $e^{(s)}$ with $e^{(t)}$, where $e^{(s)}$ and $e^{(t)}$ are the edges of the unique path from $s$ to $t$ in $T$ that are incident to $s$ and $t$, respectively.

We are now ready to prove the key lemma of this section; its proof closely follows~\cite{nutov20202nodeconnectivity}.
\begin{lemma}
\label{lemma:connected-block}
Let $T = (V_T,E_T)$ be a block-cut tree with a set of links $L = L_0 \cup L_1$. Let $H$ be the $(L,E_T)$-incidence graph. Let $J = T \cup L$. Then for $s,t\in V_T$ such that $(s,t) \notin E_T$, $\kappa_J (s,t) \geq 2$ if and only if $(s,t)$ are $H$-reachable.
\end{lemma}
\begin{proof}
Let $s \neq t\in V_T$ such that $(s,t) \notin E_T$, and let $P$ be the unique path between $s$ and $t$ in $T$. We treat $P$ both as a set of edges and a set of vertices. Let $e^{(s)}$ and $e^{(t)}$ be the edges of the unique path from $s$ to $t$ in $T$ that are incident to $s$ and $t$, respectively. We will prove the statement by using induction on the number of edges in $P$.

\paragraph{Base case.} The base case is when $P$ has $2$ edges. Let $P = (s,u,t)$. If $u = B_j$, for some $j \in [m]$, then $s$ and $t$ are necessarily cut nodes that belong to the same block $B_j$. In this case, we have $\kappa_{T \cup L_0} (s,t) \geq 2$, and by construction, we also have the path $(e^{(s)}, \ell, e^{(t)})$ in $H$, where $\ell = (e^{(s)}, e^{(t)}) \in L_0$. Thus, in this case the statement trivially holds. 

So, suppose now that $u = c \in C$. In this case, we have $s = B_i$ and $t = B_j$, $i \neq j$, such that $c \in B_i \cap B_j$. Suppose that $\kappa_J(s,t)\geq 2$. Consider $T\setminus \{c\}$ and let $T^{(s)}, T^{(t)}$ be the components of $T\setminus \{c\}$ that contain $s$ and $t$, respectively. Then, $J\setminus \{c\}$ has a path from $s$ to $t$. Consider such a path $Q$, and suppose that $Q$ goes through each of the components $T^{(s)}=T_0,T_1,\dots,T_{l-1},T_l=T^{(t)} \subset T\setminus\{c\}$ exactly once, where $l \in \mathbb{N}_{\geq 1}$; without loss of generality the components are labelled in the order $Q$ visits them. It is easy to see that we must have a link $\ell_x \in L_1$ between $T_x$ and $T_{x+1}$, for every $x \in \{0, 1, \ldots, l-1\}$. Let $e_x$ denote the edge in $T$ that is adjacent to $c$ and has an endpoint in $T_x$, for $x \in \{0, 1, \ldots, \ell\}$. Note that we have $e_0=e^{(s)}$ and $e_{l}=e^{(t)}$. It is not hard to see that we have the following path between $e^{(s)}$ and $e^{(t)}$ in $H$: $e^{(s)},\ell_0,e_1,\ell_1,\dots,\ell_{l-1},e^{(t)}$. We conclude that $(s,t)$ are $H$-reachable.
    
Suppose now that $\kappa_J(s,t)=1$. This implies that $J\setminus \{c\}$ has no path between $s$ and $t$. Let $J^{(s)}$ be the component of $J\setminus\{c\}$ that contains $s$. Let $L^{(s)}$ and $L^{(t)}$ be the set of links that have both endpoints in $J^{(s)}\cup\{c\}$ and $V\setminus J^{(s)}$, respectively. Since there is no path between $s$ and $t$ in $J\setminus \{v\}$, it is clear that $L^{(s)}$ and $L^{(t)}$ partition $L$. In particular, this means that for any link $\ell \in L^{(s)}$, the path in $T$ connecting the endpoints of $\ell$ shares no edges with the corresponding path for any link in $L^{(t)}$. Thus, there is no path between $e^{(s)}$ and $e^{(t)}$ in $H$, and the claim holds.

\paragraph{Induction step.} We now suppose that $P$ has at least $3$ edges, and that the claim holds for all paths whose number of edges is strictly smaller than the number of edges in $P$. Let $(u,v)$ and $(v,t) = e^{(t)}$ be the last two edges in $P$. In order to prove that the claim holds, we first prove several other facts we will need. From now on, let $P_{xy}$ denote the subpath of $P$ with endpoints $x \in V_T$ and $y \in V_T$. 

\begin{fact}\label{fact1}
If $\kappa_J(s,t) \geq 2$, then $\kappa_J(s,v) \geq 2$.
\end{fact}
\begin{proof}
Suppose that $\kappa_J(s,t) \geq 2$, and let $Z$ be the cycle in $J$ formed by taking the union of two node disjoint paths from $s$ to $t$ in $J$. Obviously, if $v$ is part of $Z$, we have two internally disjoint paths from $s$ to $v$ in $J$, and in this case the claim holds. If $v$ is not part of $Z$, then let $a$ be the node in $P_{sv}$ such that the path $P_{av}$ only shares a single node with the cycle $Z$; clearly, such a node $a$ always exists, as $s$ belongs both to $Z$ and $P_{sv}$. We claim that the set of edges $Z \cup P_{av} \cup \{(v,t)\}$ contains two internally disjoint paths from $s$ to $v$. In particular, the first path is the union of the shortest path $P'$ in $Z$ from $s$ to $a$ along with $P_{av}$ and the edge $(v,t)$, while the second path is a subset of the edges $Z \setminus P'$.
\end{proof}

\begin{fact}\label{fact2}
If $(s,v)$ and $(u,t)$ are both $H$-reachable, then $(s,t)$ are $H$-reachable.
\end{fact}
\begin{proof}
By assumption, there is a path between $e^{(s)}$ and $(u,v) \in E_T$ in $H$, and a path between $(u,v) \in E_T$ and $(v,t) = e^{(t)} \in E_T$. Combining these paths gives the result. 
\end{proof}

\begin{fact}\label{fact3}
If $(s,t)$ are $H$-reachable, then $(s,v)$ are $H$-reachable.
\end{fact}
\begin{proof}
Suppose that $Q$ is the path between $e^{(s)}$ and $e^{(t)}$ in $H$. If there is a link $\ell\in L$ that is part of $Q$ and is adjacent to $(u,v)$, then the claim holds. So, suppose that there is no such link. Observe that this cannot happen if $t \in C$, since in that case, $v = B_j$ and $u = c'$, which means that there exists a link $\ell' \in L_0$ that is adjacent to both $e^{(t)}$ and $(u,v)$. Thus, we must have $u = B_i$, $v = c \in C$ and $t = B_j$, for some $i \neq j \in [m]$. Let $L^{(s)}$ and $L^{(t)}$ denote the set of links in $L$ with both ends in the connected components of $T \setminus \{(u,v)\}$ containing $s$ and $t$, respectively. The only links not included in $L^{(s)} \cup L^{(t)}$ are those where the path between their endpoints in $T$ contains the edge $(u,v) \in E_T$. In particular, this implies that $Q$ only contains links from $L^{(s)} \cup L^{(t)}$. However, there is no path between $e^{(s)}$ and $e^{(t)}$ using only edges from $L^{(s)}$ and $L^{(t)}$, which contradicts the existence of the path $Q$. We conclude that the claim holds.
\end{proof}

Resuming the induction, suppose that $\kappa_J(s,t)\geq 2$. Then by Fact~\ref{fact1}, we know that $\kappa_J(s,v)\geq 2$ and $\kappa_J(u,t)\geq 2$. Using the induction hypothesis, this implies that $(s,v)$ and $(u,t)$ are $H$-reachable, and by Fact~\ref{fact2} we conclude that $(s,t)$ are $H$-reachable. 
    
Suppose now that $(s,t)$ are $H$-reachable, and assume for contradiction that $\kappa_J(s,t) =1$. By Fact~\ref{fact3}, we know that $(s,v)$ and $(u,t)$ are $H$-reachable. Using the induction hypothesis, we get that $\kappa_J(s,v) \geq 2$ and $\kappa_J(u,t)\geq 2$. Since $\kappa_J(s,t)=1$, there is a node $a \in P \setminus \{s,t\}$ such that $s$ and $t$ are in separate connected components of $J\setminus \{a\}$. If $a \in P_{su} \setminus \{s\}$, then this contradicts the fact that $\kappa_J(s,v)\geq 2$. If $a=v$, then this contradicts the fact that $\kappa_J(u,t) \geq 2$. Thus, we must have $\kappa_J(s,t) \geq 2$. This concludes the induction step, and the proof.
\end{proof}

Using the above lemma, we can show that we only need to focus on the reduced $(L,E_T)$-incidence graph.
\begin{lemma}\label{lemma:reduced-incidence}
$J$ is $2$-node-connected if and only if $H_R$ has a path between every two terminals.
\end{lemma}
\begin{proof}
We first observe that $J$ is $2$-node-connected if and only if $\kappa_J(s,t)\geq 2$ for every pair of leaves $s,t$ of $T$. Combined with Lemma~\ref{lemma:connected-block}, this implies that $J$ is $2$-node-connected if and only if for any 2 leaves $s,t\in T$, $(s,t)$ are $H$-reachable. To obtain the reduced $(L,E_G)$-incidence graph $H_R$, we perform the following steps. We first short-cut the set $E_T$ of nodes of $H$; clearly this does not modify the connectivity of $H$. Afterwards, we extend the neighborhood of each link $\ell$ of weight $1$, and make $\ell$ adjacent to any other link $\ell'$ of weight $1$ that can be reached from $\ell$ in the incidence graph via a path that only uses links of weight $0$.

We now make some observations. We observe that short-cutting ensures that if there is a path between two leaf-edges in the original incidence graph, then there is also a path in the incidence graph whose endpoints are leaf-edges and whose internal nodes are all links. Moreover, since no link of $L_0$ is adjacent to a leaf edge, this means that if there is a path between two leaf-edges in the original incidence graph, then there is also a path whose endpoints are leaf-edges and whose internal nodes are all links in $L_1$. This means that, by performing the final step where we remove the nodes of $L_0$ and $E_T \setminus R$, connectivity between leaf-edges is maintained. We conclude that for any $2$ leaves $s,t \in T$, $(s,t)$ are $H$-reachable if and only if they are $H_R$-reachable. Putting everything together, the claim holds.
\end{proof}

We now have all the ingredients necessary to provide a proof for Theorem~\ref{cor:approx-preserve-block}.

\begin{proof}[Proof of Theorem~\ref{cor:approx-preserve-block}]
Given a connected graph $G = (V, E)$ with a set of links $L$, we construct the block-cut tree $T = (V_T, E_T)$ and the reduced $(f_G(L) \cup L_0, E_T)$-incidence graph $H_R$. The resulting instance is viewed as a Node Steiner Tree instance, where $R$, the set of leaf-edges of $T$, is the set of terminals and $f_G(L)$ is the set of links. Lemmas~\ref{lemma:reduced-incidence} and~\ref{lemma:reduction-to-block-cut} imply that $J' = G\cup L'$ is 2-node-connected, for some $L' \subseteq L$, if and only if the set $f_G(L')$ is a feasible solution for the constructed Node Steiner Tree $H_R$, and moreover, the cost of $L'$ in both instances is the same and equal to $|L'|$. 

It remains to see that the resulting Node Steiner Tree instance is in fact a CA-Node-Steiner-Tree instance, which can be done by verifying that the requirements for a CA-Node-Steiner-Tree instance hold. By definition of the incidence graph it is clear that terminals are only adjacent to Steiner Nodes. By definition of the short-cut $(L,E_G)$-incidence graph we know that the neighbors of a terminal form a clique. Finally, any link $\ell$ in $H$ is adjacent only to the edges that appear on the path in $T$ between its endpoints, and for any such path at most two edges are leaf-edges. Since we only have leaf-edges in the reduced incidence graph, we get that each Steiner node is adjacent to at most two terminals. This concludes the proof.
\end{proof}